\ifx\SOCG\undefined%
\documentclass[12pt]{article}%
\else%
\documentclass[11pt]{article}%
\fi

\usepackage[breaklinks]{hyperref} 
\usepackage{graphicx}
\usepackage[active]{srcltx} 
\usepackage{euscript}%
\usepackage{xspace}%
\usepackage{color}%
\usepackage{amsmath,amsfonts}%
\usepackage{paralist}%
\usepackage{theorem}%
\usepackage{picins}%
\usepackage{boxedminipage}
\usepackage{enumerate}
\usepackage{picins}
\usepackage{multirow}

\ifx\SOCG\undefined
\newcommand{\InConfVer}[1]{}
\newcommand{\InFullVer}[1]{#1}
\else
\newcommand{\InConfVer}[1]{#1}
\newcommand{\InFullVer}[1]{}
\fi

\InConfVer{%
   \def\ProofInAppendix{TRUE}%
}%
\usepackage{proofapnd}%

\newcommand{\xparagraph}[1]{{\smallskip\textsf{#1}}}

\newcommand{\Opt}{\mathrm{opt}}%
\newcommand{\OptLP}{\mathrm{opt}_{\text{LP}}}
\newcommand{\PTAS}{\Term{P{T}AS}\xspace}

\newcommand{\ts}{\hspace{0.6pt}}%
\renewcommand{\th}{\si{th}\xspace}

\newcommand{\scl}{\Delta}%
\newcommand{\sclC}{\alpha}

\providecommand{\si}[1]{#1}

\newcommand{\AlinaThanks}[1]{%
   \thanks{Department of Computer Science; %
      University of Illinois; %
      201 N. Goodwin Avenue; %
      Urbana, IL, 61801, USA; %
      {\tt \si{ene}1\atgen{}\si{uiuc}.\si{edu}}; {\tt
         \si{\url{http://www.cs.uiuc.edu/\string~ene1/}}.} #1}}

\newcommand{\BenThanks}[1]{\thanks{Department of Computer Science;
      University of Illinois; %
      201 N. Goodwin Avenue; %
      Urbana, IL, 61801, USA; %
      {\tt \si{raichel}2\atgen{}\si{uiuc}.\si{edu}}; %
      {\tt \url{\si{http://www.cs.uiuc.edu/\string~\si{raichel2}}}}%
      . %
      #1}}

\newcommand{\atgen}{\symbol{'100}}
\newcommand{\SarielThanks}[1]{\thanks{Department of Computer Science;
      University of Illinois; 201 N. Goodwin Avenue; Urbana, IL,
      61801, USA; {\tt \si{sariel}\atgen{}\si{uiuc.edu}}; %
      {\tt \url{http://www.uiuc.edu/\string~sariel/}.} #1}}

\definecolor{blue25}{rgb}{0,0,0.55}%
\newcommand{\emphic}[2]{%
   \textcolor{blue25}{%
      \textbf{\emph{#1}}}%
   \index{#2}}

\newcommand{\emphi}[1]{\emphic{#1}{#1}}
\newcommand{\pth}[2][\!]{#1\left({#2}\right)}

\definecolor{red25}{rgb}{0.4,0,0.0}%

\newcommand{\PStyle}[1]{\textcolor{red25}{\textsc{#1}}}
\newcommand{\PackRegions}       {\PStyle{{Pack{}Regions}}\xspace}
\newcommand{\PackPoints}        {\PStyle{{Pack{}Points}}\xspace}
\newcommand{\PackHGraph}        {\PStyle{{HGraph{}Packing}}\xspace}%
\newcommand{\PackHalfspaces}    {\PStyle{{Pack{}Halfspaces}}\xspace}
\newcommand{\PackRaysInPlanes}  {\PStyle{{Pack{}Rays{}In{}Planes{}}}\xspace}
\newcommand{\PackPointsInDisks} {\PStyle{{Pack{}Points{}In{}Disks}}\xspace}
\newcommand{\PackRectsInPnts}   {\PStyle{{Pack{}Re{}ct{}s{}In{}Points}}\xspace}
\newcommand{\PackBoxesInPnts}   {\PStyle{{Pack{}Boxes{}In{}Points}}\xspace}
\newcommand{\PackPntsInSkyline} {\PStyle{{Pack{}Pnts{}In{}Skyline}}\xspace}
\newcommand{\PackPntsInRects}   {\PStyle{\si{PackPntsInRects}}\xspace}
\newcommand{\PackPntsInFTri}    {\PStyle{\si{PackPntsInFatTriangs}}\xspace}

\newtheorem{theorem}{Theorem}[section]
\newtheorem{lemma}[theorem]{Lemma}%
\newtheorem{rexample}[theorem]{Running Example}
\newtheorem{definition}[theorem]{Definition}

\newtheorem{corollary}[theorem]{Corollary}
\newtheorem{problem}[theorem]{Problem}%

{\theorembodyfont{\rm} \newtheorem{remark}[theorem]{Remark}}

\newcommand{\brc}[1]{\left\{ {#1} \right\}}

\newcommand{\ObjSet}{\EuScript{D}}%

\newcommand{\PlaneSet}{\EuScript{H}}%
\newcommand{\plane}{\mathsf{h}}%

\newcommand{\HalfspaceSet}{\EuScript{S}}%
\newcommand{\RaySet}{\EuScript{R}}%

\newcommand{\RectSet}{\EuScript{B}}
\newcommand{\RectSetA}{\EuScript{D}}

\newcommand{\BoxSet}{\EuScript{B}}
\newcommand{\rect}{\mathsf{b}}

\newcommand{\VSetA}{X}%
\newcommand{\VSetB}{Y}%
\newcommand{\VSetD}{U}%
\newcommand{\VOpt}{V_{\Opt}}%

\newcommand{\Term}[1]{\textsf{#1}}%

\newcommand{\VC}{\Term{VC}\xspace}
\newcommand{\LP}{\Term{LP}\xspace}%
\newcommand{\LCA}{\Term{LCA}\xspace}%
\newcommand{\HyperLP}{\textsc{Hypergraph-LP}\xspace}
\newcommand{\APX}{\Term{AP{X}}\xspace}%
\newcommand{\APXHard}{\Term{AP{X}-hard}\xspace}%
\newcommand{\NPHard}{\Term{NP-Hard}\xspace}%
\renewcommand{\P}{\Term{P}\xspace}%
\newcommand{\NP}{\Term{NP}\xspace}%
\newcommand{\hedge}{f}%
\newcommand{\hedgeA}{z}%

\newcommand{\conflict}{h}%
\newcommand{\conflictSet}{\mathcal{H}}%
\newcommand{\RSample}{\mathsf{R}}%

\newcommand{\graph}{{G}}%
\newcommand{\hgraph}{\mathsf{G}}%
\newcommand{\VSet}{\mathsf{V}}%
\newcommand{\HESet}{\mathsf{E}}%

\newcommand{\vertex}{v}%
\newcommand{\weight}[1]{w\pth{#1}}%

\newcommand{\RSet}{\mathcal{C}}%
\newcommand{\OSet}{\mathcal{O}}

\newcommand{\Energy}{\EuScript{E}}

\newcommand{\EnergyX}[2][\!]{\EuScript{E}\pth[#1]{#2}}

\newcommand{\cRelax}{\phi}
\newcommand{\forceX}[2][\!]{\rho\pth[#1]{#2}}
\newcommand{\forceY}[2]{\rho_{#1}\pth{#2}}%
\newcommand{\resistC}{\eta}%

\newcommand{\resistZ}[3]{\resistC_{#1}\pth{#2, #3}}
\newcommand{\Union}[2][\!]{\mathsf{U}\pth[#1]{#2}}
\newcommand{\union}[2][\!]{\mathsf{u}\pth[#1]{#2}}
\renewcommand{\Re}{{\rm I\!\hspace{-0.025em} R}}

\newcommand{\mc}{\nu}

\newcommand{\mbb}     {\rule[0.0cm]{0.0cm}{0.38cm}} 
\newcommand{\MakeBig} {\rule[-.2cm]{0cm}{0.4cm}}
\newcommand{\MakeSBig}{\rule[0.0cm]{0.0cm}{0.38cm}} 

\newcommand{\seclab}[1]{\label{sec:#1}}
\newcommand{\secref}[1]{Section~\ref{sec:#1}}

\newcommand{\problab}[1]{\label{prob:#1}}
\newcommand{\probref}[1]{Problem~\ref{prob:#1}}

\newcommand{\thmlab}[1]{{\label{theo:#1}}}%
\newcommand{\thmref}[1]{Theorem~\ref{theo:#1}}%

\providecommand{\deflab}[1]{\label{def:#1}}
\newcommand{\defref}[1]{Definition~\ref{def:#1}}

\newcommand{\corlab}[1]{\label{cor:#1}}
\newcommand{\corref}[1]{Corollary~\ref{cor:#1}}

\newcommand{\itemlab}[1]{\label{item:#1}}
\newcommand{\itemref}[1]{(\ref{item:#1})}

\newcommand{\lemlab}[1]{\label{lemma:#1}}
\newcommand{\lemref}[1]{Lemma~\ref{lemma:#1}}

\newcommand{\remlab}[1]{\label{rem:#1}}
\newcommand{\remref}[1]{Remark~\ref{rem:#1}}

\newcommand{\figlab}[1]{\label{figure:#1}}
\newcommand{\figref}[1]{Figure~\ref{figure:#1}}

\newcommand{\TwoFigures}[6]{\begin{tabular}{cc}
       \begin{minipage}{0.48\linewidth}
           \begin{center}
               \centerline{\includegraphics[#1]{{#2}}}
           \end{center}
       \end{minipage}
       &
       \begin{minipage}{0.48\linewidth}
           \centerline{\includegraphics[#4]{{#5}}}
       \end{minipage}
       ~\\[-0.2cm]
       {#3} & {#6}
   \end{tabular}}

\newenvironment{proof}%
  {\trivlist\item[]\emph{Proof}:}%
  {\unskip\nobreak\hskip 1em plus 1fil\nobreak%
   \myqedsymbol
   \parfillskip=0pt%
   \endtrivlist}
  {\trivlist\item[]\emph{\textbf{Proof of #1}}:}%
  {\unskip\nobreak\hskip 1em plus 1fil\nobreak%
   \myqedsymbol
   \parfillskip=0pt%
   \endtrivlist}

\newcommand{\myqedsymbol}{\rule{2mm}{2mm}}
\newcommand{\cardin}[1]{\left| {#1} \right|}
\newcommand{\pbrcx}[1]{\left[ {#1} \right]}
\newcommand{\Ex}[2][\!]{\mathop{\mathbf{E}}#1\pbrcx{#2}}
\newcommand{\Prob}[1]{\mathop{\mathbf{Pr}}\!\pbrcx{#1}}
\newcommand{\sep}[1]{\,\left|\, {#1} \MakeBig\right.}
\newcommand{\permut}[1]{\left\langle {#1} \right\rangle}%

\newcommand{\capacityX}[1]{\#(#1)}

\newcommand{\Pnum}{\psi}%
\newcommand{\eqlab}[1]{\label{equation:#1}}
\newcommand{\Eqref}[1]{\si{Eq}.~(\ref{equation:#1})}
\providecommand{\ds}{\displaystyle}

\newcommand{\local}{\mathsf{L}}%
\newcommand{\eps}{\varepsilon}%

\newcommand{\seg}{\mathsf{s}}%
\newcommand{\SegSet}{\mathsf{S}}%
\newcommand{\feq}[2]{F_{{#1}} \pth{#2}}%

\newcommand{\FamilyA}{\mathcal{Z}}
\newcommand{\InducedX}[2]{#1_{#2}}
\newcommand{\Alg}{\ensuremath{\mathtt{alg}}\xspace}
\newcommand{\iright}{i_{\mathrm{right}}}
\newcommand{\ileft}{i_{\mathrm{left}}}
\newcommand{\remove}[1]{}
\newcommand{\xkappa}{\EuScript{M}}

\newcommand{\ceil}[1]{\left\lceil {#1} \right\rceil}
\newcommand{\floor}[1]{\left\lfloor {#1} \right\rfloor}
\newcommand{\Turan}{T{u}r\'an\xspace}

\newcommand{\etal}{\textit{et~al.}\xspace}

\remove{
\usepackage{pifont}%

\usepackage[symbol*]{footmisc}

\DefineFNsymbols*{stars}[text]{%
   {\ding{172}} %
   {\ding{173}} %
   {\ding{174}} %
   {\ding{175}} %
   {\ding{176}} %
   {\ding{177}} %
   {\ding{178}} %
   {\ding{179}} %
   {\ding{180}} %
   }
\setfnsymbol{stars}
}

\newcommand{\DLeftX}[1]{#1_{\leftarrow}}
\newcommand{\DLeftTopX}[1]{#1_{\nwarrow}}
\newcommand{\DLeftBottomX}[1]{#1_{\swarrow}}
\newcommand{\DRightX}[1]{#1_{\rightarrow}}
\newcommand{\Tri}{\triangle}%
\newcommand{\TriL}{\DLeftX{\Tri}}%
\newcommand{\TriLT}{\DLeftTopX{\Tri}}%
\newcommand{\TriLB}{\DLeftBottomX{\Tri}}%
\newcommand{\TriR}{\DRightX{\Tri}}%

\newcommand{\TriSet}{\EuScript{T}}%
\newcommand{\TriSetA}{\EuScript{S}}%
\newcommand{\TriSetLB}{\DLeftBottomX{\TriSet}}%
\newcommand{\TriSetLT}{\DLeftTopX{\TriSet}}
\newcommand{\TriSetR} {\DRightX{\TriSet}}

\newcommand{\pnt}{{\mathsf{p}}}%
\newcommand{\pntA}{{\mathsf{q}}}%

\newcommand{\Line}{\ell}
\newcommand{\LineA}{\wp}

\newcommand{\PntSet}{\mathsf{P}}%

\newcommand{\regionA}{r}

\newcommand{\FTriConst}{9}

\newcommand{\apndlab}[1]{\label{apnd:#1}}
\newcommand{\apndref}[1]{Appendix~\ref{apnd:#1}}

\newcommand{\ParPicBeforeParagraph}[1]{ \vspace{0.8cm}

   #1

   \vspace{-0.8cm} }

\InConfVer{%
}

\InConfVer{ \newcommand{\secapndref}[1]{Appendix~\ref{sec:#1}}}
\InFullVer{ \newcommand{\secapndref}[1]{Section~\ref{sec:#1}}}

\begin{document}

\title{Geometric Packing under Non-uniform Constraints%
   \footnote{The full version of the paper is available from the
      \si{arxiv} \cite{ehr-gpnuc-11}.}  }

\author{%
   Alina Ene%
   \AlinaThanks{Work on this paper was partially supported by NSF
      grants CCF-0728782 and CCF-1016684.}%
   \and%
   Sariel Har-Peled%
   \SarielThanks{Work on this paper was partially supported by a NSF
      AF award CCF-0915984.}%
   \and%
   Benjamin Raichel%
   \BenThanks{Work on this paper was partially supported by a NSF AF
      award CCF-0915984.}%
}

\date{\today}

\maketitle
\InConfVer{%
   \setcounter{page}{0}

   \thispagestyle{empty}}

\addtocounter{footnote}{3}


\begin{abstract}
    We study the problem of discrete geometric packing. Here, given
    weighted regions (say in the plane) and points (with capacities),
    one has to pick a maximum weight subset of the regions such that
    no point is covered more than its capacity. We provide a general
    framework and an algorithm for approximating the optimal solution
    for packing in hypergraphs arising out of such geometric
    settings. Using this framework we get a flotilla of results on
    this problem (and also on its dual, where one wants to pick a
    maximum weight subset of the points when the regions have
    capacities). For example, for the case of fat triangles of similar
    size, we show an $O(1)$-approximation and prove that no \PTAS is
    possible.
\end{abstract}

\InConfVer{\newpage}

\section{Introduction}
\seclab{intro}

\paragraph{Motivation and examples.}

Consider the problem of \emphi{obnoxious facility location}
\cite{t-oflg-91, C-sofl-99}; that is, you have to place several
facilities, but these facilities are undesired (i.e.,
obnoxious). Facilities of this type include nuclear reactors, wind
farms, airports, power plants, factories, prisons, universities,
etc. Facilities can also be semi-desirable -- a customer might want to
have supermarkets close to their home, but they do not want to have
too many of them close by as they increase traffic, noise, etc. One
natural way to model this geometrically is to associate each obnoxious
facility with its region of undesirability. We also have customers
(modeled as points), and each customer has a threshold of how many
obnoxious facilities it is willing to accept covering it. Different
customers may have different thresholds, for example because more
affluent people have stronger political power and it is harder to
place obnoxious facilities near their homes.

Naturally, if you allow only a single region to cover each customer,
then this is a classical packing problem, and much work has been done
on packing disks/balls \cite{smcscg-nacps-07}. However, there are many
cases where allowing limited interaction between the packed regions is
allowed (after all, these facilities are required for modern existence).
As a concrete example of this type of problem, consider the placement
of radio stations/cellphone towers. While airports allow only very
limited levels of interference\footnote{See
   \si{\url{\si{http://tinyurl.com/7td67v3}}} for a story of an
   airport closing down because of radio interference.}, higher levels 
of such interference is acceptable in residential neighborhoods.  However, 
at a certain point there is going to be resistance to placing 
more wireless towers in residential areas, as these towers are viewed as 
causing cancer (this fear might be baseless, but it does not change the political reality of the difficulty of placing such towers). On the other 
hand, there is little resistance to placing such towers along highways in sparsely populated areas.

In this paper, we are interested in the modeling of such problems and
in the computation of an efficient approximation to the optimal
solution of such problems.

\subsection*{Modeling.}
As hinted by the above, perhaps the most natural way to model this problem is 
as a generalization of the well known independent set problem.

\paragraph{Independent set} %
is a fundamental discrete optimization problem.  Unfortunately, it is
not only computationally hard, but it is even hard to approximate to
within a factor of $n^{1-\eps}$, for any constant $\eps$
\cite{h-chaw-99} (under the assumption that
$\textsf{NP}\neq\textsf{P}$). Surprisingly, the problem is
considerably easier in some geometric settings. For example, there is
a \PTAS\footnote{Polynomial time approximation scheme.}
\cite{c-ptasp-03,ejs-ptasg-05} for the following problem: Given a set
of unit disks in the plane, find a maximum cardinality subset of the
disks whose interiors are disjoint.  Furthermore, a simple local
search algorithm yields the desired approximation: For any $\eps > 0$,
the local search algorithm that tries to swap subsets of size $O(1 /
\eps^2)$ yields a $(1 - \eps)$-approximation in $n^{O(1/\eps^2)}$ time
\cite{ch-aamis-09, ch-aamis-11}.

\paragraph{The discrete independent set problem.} %
In this paper, we consider packing problems in geometric settings that
are natural extensions of the geometric independent set problem
described above.  As a starting point, motivated by practical
applications, we consider the discrete version of the geometric
independent set problem in which, in addition to a set of weighted
regions, we are given a set of points, and the goal is to select a
maximum weight subset of the regions so that each point is contained
in at most one of the selected regions.  We refer to this problem as
the \emphi{discrete independent set} problem.  Chan and Har-Peled
\cite{ch-aamis-11} studied this discrete variant and proved that one
can get a good approximation if the union complexity of the regions is
small.

\parpic[r]{\includegraphics{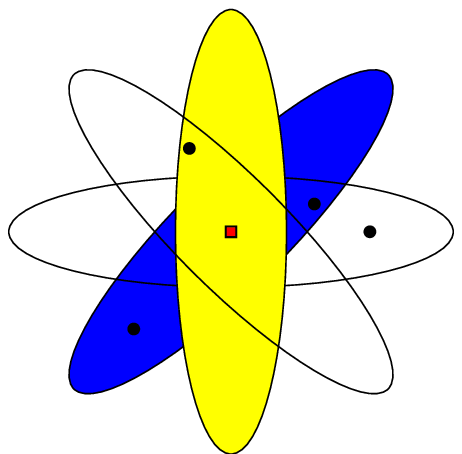}}%

Note that the discrete independent set problem captures the continuous
version of the independent set problem, since we can place a point in
each face of the induced arrangement of the given regions.  In fact,
the discrete version is considerably harder (in some cases) than the
continuous variant.  The difficulty lies in that several regions
forming a valid solution to an instance of a discrete independent set
problem may contain a common point that is not part of the set of
points given as input; the figure on the right shows an example in
which the middle point, marked as a square, is covered twice by the
given valid solution.

\parpic[r]{\includegraphics{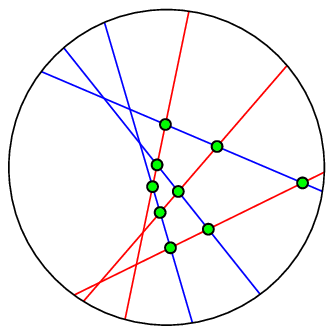}}%

To illustrate the difference in difficulty, consider the case when the
input consists of a set $\SegSet$ of segments (in general position)
with their endpoints on a circle, such that every pair of segments
intersect.  Clearly, in the continuous version, the maximum
independent set of segments is a single segment.  However, in this
case, the discrete version captures the graph independent set
problem. More precisely, we can encode any instance of independent set
(i.e., a graph $G=(V,E)$) as an instance of this problem as
follows. Every vertex $v \in V$ is mapped to a segment $\seg_v$ of
$\SegSet$, and every edge $uv \in E$, is mapped to the point $\seg_u
\cap \seg_v$ (which is added to a set of points $\PntSet$). Clearly,
an independent set of segments of $\SegSet$ (in relation to the point
set $\PntSet$) corresponds to an independent set in $G$. That is, the
geometric discrete version is sometimes as hard as the graph
independent set problem.  For example, the figure on the right depicts
the resulting instance encoding independent set for $K_{3,3}$.

\paragraph{The packing problem.} %
In this paper, we are interested in the natural extension of the
discrete independent set problem to the case where every point has a
capacity and might be covered several times (but not exceeding its
capacity). The resulting problem has a flavor of a packing problem,
and is defined formally as follows.

\begin{problem} {\rm{(\PackRegions.)}} %
    Given a set $\ObjSet$ of regions and a set $\PntSet$ of points
    such that each region $r$ has a weight $\weight{r}$ and each point
    $\pnt$ has a capacity $\capacityX{\pnt}$, find a maximum weight
    subset $X$ of the regions such that, for each point $\pnt$, the
    number of regions in $X$ that contain $\pnt$ is at most its
    capacity $\capacityX{\pnt}$.
    
    \problab{pack:regions}
\end{problem}

We emphasize that different points might have different capacities,
which makes the problem considerably more challenging to solve than
the unit capacities case (i.e., the discrete independent set problem).
We also consider the following dual problem in which the points have
weights and the regions have capacities.

\begin{problem}{\rm{(\PackPoints.)}}
    Given a set $\ObjSet$ of regions and a set $\PntSet$ of points
    such that each region $r$ has a capacity $\capacityX{r}$ and each
    point $\pnt$ has a weight $\weight{\pnt}$, find a maximum weight
    subset $X$ of the points such that each region $r$ contains at
    most $\capacityX{r}$ points of $X$.
    
    \problab{pack:points}
\end{problem}

\noindent
\textbf{Hypergraph framework.} %
These two problems can be stated in a unified way in the language of
hypergraphs\footnote{A \emphi{hypergraph} $\hgraph$ is a pair
   $(\VSet,\HESet)$, where $\VSet$ is a set of vertices and $\HESet$
   is a collection of subsets of $\VSet$ which are called
   \emphi{hyperedges}.}. Given an instance of \PackRegions, we
construct a hypergraph as follows: Each weighted region is a vertex,
and all the regions containing a given point of capacity $k$ become a
hyperedge (consisting of these regions) of capacity $k$.  A similar
reduction works for \PackPoints, where the given weighted points are
the vertices, and each region of capacity $k$ becomes a hyperedge of
capacity $k$ consisting of all of the points contained in this
region. Therefore the previous two problems are special cases of the
following problem.

\begin{problem}{\rm{(\PackHGraph.)}}
    Given a hypergraph $\hgraph = (\VSet,\HESet)$ with a weight
    function $\weight{\cdot}$ on the vertices and a capacity function
    $\capacityX{\cdot}$ on the hyperedges, find a maximum weight
    subset $X\subseteq \VSet$, such that $\forall \hedge\in \HESet$ we
    have $|X\cap \hedge|\leq \capacityX{\hedge}$.
\end{problem}

We will be interested primarily in hypergraphs with certain hereditary
properties. A hypergraph property is \emph{hereditary} if the
sub-hypergraph induced by any subset of the vertices has the property;
an example of a hereditary property of hypergraphs is having bounded
\VC dimension. Roughly, we are interested in hypergraphs having the
\emphi{bounded growth property}: For any induced sub-hypergraph on $t$
vertices the number of its hyperedges that contain exactly $k$
vertices is near linear in $t$ and its dependency on $k$ is bounded by
$2^{O(k)}$, see \defref{bounded}. Such hypergraphs arise naturally
when considering points and ``nice'' regions in the plane.


\subsection*{Our results.}
\begin{compactitem}[$\bullet$]
    \item \xparagraph{Main result.} Our main result is an algorithm
    that provides a good approximation for \PackHGraph as a function
    of the growth of the hypergraph, see \thmref{p:hypergraph}.  Our
    result can be viewed as an extension of the work of Chan and
    Har-Peled \cite{ch-aamis-11} to these considerably more general
    and intricate settings.
    
    \item \xparagraph{Regions with low union complexity.}  In
    \secapndref{applications}, we apply our main result to regions that
    have low union complexity, and we get the following results:
    
    \begin{compactenum}[(A)]
        \item If the union complexity of $n$ regions is $O( n u(n))$
        then we get an $O\pth{ u(n)^{1/\mc} }$-approximation for
        \PackRegions, where $\mc$ is the minimum capacity of any point
        in the given instance. (That is, the problem becomes easier as
        the minimum capacity increases.)  For the case where all the
        capacities are one, this is the discrete independent set
        problem, and our algorithm specializes to the algorithm of
        Chan and Har-Peled \cite{ch-aamis-11}, which gives an
        $O(u(n))$-approximation.
        
        \item More specifically, we get a constant factor approximation 
        for \PackRegions if the union complexity of the regions is
        linear. This holds for
        \begin{inparaenum}[(i)]
            \item fat-triangles of similar size,
            \item unit axis-parallel cubes in 3d, and
            \item pseudo-disks.
        \end{inparaenum}
        See \corref{pseudo}.
        
        \item Similarly, since the union complexity of fat triangles 
        in the plane is $O(n \log^* n )$ \cite{eas-ibuft-11, abes-ibufop-11}, 
        we get an $O\pth{ \pth{\log^* n}^{1/\mc} }$ approximation for
        such instances of \PackRegions.
    \end{compactenum}
    
    \item \xparagraph{Bi-criteria approximation.} %
    Our main result also implies a bi-criteria approximation
    algorithm. That is, we can improve the quality of the solution, at
    the cost of potentially violating low capacity regions.  Formally,
    if the input instance $\hgraph = (\VSet, \HESet)$ of \PackHGraph
    has at most $\feq{k}{t} = 2^{O(k)} F(t)$ edges of size $k$ when
    restricted to any subset of $t$ vertices, then for any integer
    $\cRelax \geq 1$, our algorithm yields an $\pth{
       O\pth{(F(n)/n)^{1/\cRelax}}, \cRelax} $-approximation to the
    given instance $\hgraph$ of \PackHGraph.  Specifically, the value
    of the generated solution $\VSetA$ is at least $\Omega\pth{ \Opt/
       (F(n)/n)^{1/\cRelax}}$, where $\Opt$ is the value of the
    optimal solution, and for every hyperedge $\hedge \in \HESet$, we
    have $\cardin{\hedge \cap \VSetA} \leq \max \pth{ \cRelax,
       \capacityX{\hedge}}$.
    
    As an example, for any set of $n$ regions in the plane such that
    the boundaries of any pair of them intersects $O(1)$ times, the
    above implies that one can get an $\pth{O\pth{n^{1/\cRelax}}, \,
       \cRelax}$-approximation for \PackRegions.
    
    \item \xparagraph{Axis-parallel boxes.} The union complexity of
    axis-parallel rectangles can be as high as quadratic, and
    therefore we cannot immediately apply our main result to get a
    good approximation. Instead, we decompose the union of
    axis-parallel rectangles into regions of low union complexity, and
    this decomposition together with our main result gives us an
    $O(\log n)$ approximation for instances of \PackRegions in which
    the regions are axis-parallel rectangles in the plane%
    (see \lemref{p:rect:in:pnts}).
    
    A more involved analysis also applies to the three dimensional
    case, where we get an $O(\log^3n)$ approximation for \PackRegions
    for axis parallel boxes (see \lemref{p:boxes:in:pnts}).
    
    \item \xparagraph{Dual problem.}  We show in
    \secapndref{halfRayDisk} that, by standard lifting techniques, we
    can apply our result for \PackRegions, where the regions are
    disks, to the dual problem of \PackPointsInDisks.  However, for
    other regions, the dual problem \PackPoints seems to be more
    challenging.  Specifically, this is true for the case of
    axis-parallel rectangles.  For this case, we first provide a
    constant factor approximation for \emph{skyline} instances of the
    problem; a skyline instance is a set of rectangles that lie on the
    $x$-axis. Interestingly, if the set of rectangles are defined in
    relation to a set of points (and each rectangles contains only a
    few points), then one can define a near-linear (in the number of
    points) sized set of rectangles such that each original rectangle
    is the union of two new rectangles. Combining this with the
    skyline result and a sparsifying technique, we get an $\pth{ O(
       \log n), 2}$-approximation; that is, every rectangle $\rect$
    contains at most $\max(2, \capacityX{\rect})$ points of the
    solution constructed, and the total weight of the solution is
    $\Omega(\Opt/\log n)$%
    (see \thmref{pack:points:in:rects}). (Note that, by applying our
    general framework directly to this setting, we only get an $\pth{
       O\pth{n^{1/\cRelax}}, \cRelax}$-approximation, for any integer
    $\cRelax > 0$.)
    
    \item \xparagraph{Packing points into fat triangles.}  %
    We provide a polylog bi-criteria approximation for the problem of
    packing points into fat triangles. This requires proving that one
    can compute, for a given point set, a small number of canonical
    subsets, such that the point set covered by any fat-triangle (if
    the set is sufficiently small), is the union of constant number of
    these canonical subsets. Proving this requires non-trivial
    modifications of the result of Aronov \etal
    \cite{aes-ssena-10}. In addition, we show that a measure defined
    over a fat triangle can be covered by a few fat triangles, each
    one of them containing only a constant fraction of the original
    measure. We believe these two results are of independent
    interest. Plugging this into the machinery, previously developed
    for axis parallel rectangles, yields the new approximation
    algorithm. See \secapndref{p:points:in:triangles} for details.
    
    
    \item \xparagraph{\PTAS for disks and planes.} We adapt the
    techniques of Mustafa and Ray \cite{mr-irghs-10} in order to get a
    PTAS for instances consisting of unweighted disks and
    unit-capacity points: we lift the problem to 3d, we construct an
    approximate conflict graph (as done by Mustafa and Ray), and we
    use a local search algorithm. This result also implies a \PTAS for
    \PackPoints for unweighted points and uniform capacity halfspaces
    in $\Re^3$.  See \secapndref{PTAS} for the details.
    
    \item \xparagraph{Hardness.} %
    We show some hardness results for our problems. In particular, we
    show that \PackPoints for fat triangles in the plane is as hard as
    independent set in general graphs (see
    \lemref{l:b:pack:pnts:into:triangles}).  We also show that
    \PackRegions is \APXHard (and thus there is no \PTAS) for
    similarly sized fat triangles in the plane (thus ``matching'' the
    result of \corref{pseudo}).
\end{compactitem}

\paragraph{Main technical contribution.}
Besides the results mentioned above, our work further develops and
extends the techniques for rounding \LP{}s that rise out of low
dimensional geometric constraints. Such work relies on finding the
right order of making decisions about regions as they are being added,
usually initially picking the elements to be considered randomly
according to the value assigned to them by the associated \LP. Such
work in the context of \LP rounding in geometric settings includes
\cite{v-wgscq-10, ch-aamis-09, cch-smcpg-09}.  The basic idea is to
build a conflict graph, on the appropriate random sample, and argue
that there exists a vertex of low degree that can be added without
throwing away too many conflicting vertices.  Our work extends this
approach to more involved settings where conflicts are not just
whether two regions intersect or not (i.e., independent set in a
graph), but rather involve a larger number of regions. To this end, we
prove a combinatorial bound on the expected number of conflicts
realized if we round the associated \LP.  A special easier case of
this was addressed by Chan and Har-Peled \cite{ch-aamis-09} when the
\LP is an independent set \LP. Naturally, in our case the analysis is
considerably more involved.

\paragraph{Previous work.} %
Fox and Pach \cite{fp-cinig-11} presented an $n^{\eps}$ approximation
for independent set for segments in the plane. The usage of \LP
relaxations for approximating such problems is becoming more popular.
In particular, Chalermsook and Chuzhoy \cite{cc-misr-09} use a natural
LP relaxation to get an $O(\log \log m)$-approximation for independent
set of axis parallel rectangles in the plane. The geometric set cover
problem and the more general problem, the geometric set multi-cover
problem, have approximation algorithms that use $\eps$-nets to round
the natural \LP relaxation; see \cite{cch-smcpg-09} and references
therein. Chan and Har-Peled \cite{ch-aamis-09} used local search to
get a \PTAS for independent set of pseudo-disks.  Independently,
Mustafa and Ray \cite{mr-irghs-10} used similar ideas to get a \PTAS
for hitting set of pseudo-disks in the plane. There is not much work
on the hardness of optimization problems in the geometric settings we
are interested in.  \cite{cc-ccopd-07} shows that the problem of
independent set of axis-parallel boxes in three dimensions is \APXHard
(the problem is known to be \NPHard in the plane).  See also
\cite{gc-eaahr-11, h-bffne-09} and references therein for some recent
hardness results.  Naturally, in non-geometric settings, there is a
vast literature on the problems and techniques we use, see
\cite{ws-daa-11}. Our algorithms use the randomized rounding with
alteration technique to round a fractional solution. This technique
was used in \cite{s-nacpp-01} to find an approximate solution to
packing integer programs (PIPs) of the form $\{\max w x : Ax \leq b, x
\in \mathbb{Z}_+^{n}\}$, where $A$ is a matrix whose entries are
either $0$ or $1$. The approximation guarantee given in
\cite{s-nacpp-01} is $O(n^{1 / B})$, where $B = \min_i b_i$.

\paragraph{Organization.}
In \secref{prelims} we define the problem and the associated \LP
relaxation, and describe some basic tools used throughout the
paper. In \secref{hypergraph} we present the approximation algorithm
for the hypergraph case. %
\InConfVer{We conclude in \secref{conclusions}. Due to minor space
   limitations, we move many of the technical parts of the paper to
   the appendices. } %
In \secapndref{applications} we present various applications of our
main result.  In \secapndref{p:points:in:triangles} we present the
algorithm for packing points into fat triangles.  In \secapndref{PTAS}
we present a \PTAS for some restricted cases.  In
\secapndref{hardness} we present some hardness results.  \InFullVer{
   We conclude in \secref{conclusions}.}

\section{Preliminaries}
\seclab{prelims}

For a maximization problem, an algorithm provides an
\emphi{$\alpha$-approximation} if it outputs a solution of value at
least $\Opt/\alpha$, where $\Opt$ is the value of the optimal
solution. An \emphi{$(\alpha, \beta)$-approximation} algorithm for
\PackHGraph is an algorithm that returns a (potentially infeasible)
solution of value at least $\Opt/\alpha$ such that each hyperedge
$\hedge$ contains at most $\max( \capacityX{\hedge}, \beta)$ vertices
of the solution.

\InFullVer{%
\paragraph{$\alpha$-fat triangles.}
For $\alpha \geq 1$, a triangle $\triangle$ is \emphi{$\alpha$-fat} if
the ratio between its longest edge and its height on this edge is
bounded by $\alpha$ (there are several equivalent definitions of this
concept). A set of triangles is $\alpha$-fat if all the triangles in
the set are $\alpha$-fat.  The union complexity of $n$ $\alpha$-fat
triangles is $O(n \log^* n )$ \cite{eas-ibuft-11, abes-ibufop-11} (the
constant in the $O$ depends on $\alpha$, which is assumed to be a
constant).  }

\subsection{\LP Relaxation and the Rounding Scheme}
\seclab{relax}

We consider the following natural \LP relaxation for the \PackHGraph
problem. For each vertex $v$, we have a variable $x_v$ with the
interpretation that $x_v$ is $1$ if $v$ is selected, and $0$
otherwise. For each hyperedge $\hedge$, we have a constraint that
enforces that the number of vertices of $\hedge$ that are selected is
at most the capacity of $\hedge$.

\centerline{%
   \begin{minipage}{0.7\textwidth}%
       \begin{align*}
           \HyperLP: \qquad \max \quad & \sum_{v \in \VSet} w_v x_v\\
           &\sum_{v \in \hedge} x_v \leq \capacityX{\hedge} \qquad &\forall
           \hedge
           \in \HESet\\
           & 0 \leq x_v \leq 1 \qquad &\forall v \in \VSet.
       \end{align*}
   \end{minipage}%
}

\bigskip
\noindent
The \emphi{energy} of a subset $\VSetA \subseteq \VSet$ is $\ds
\EnergyX{\VSetA} = \sum_{\vertex \in \VSetA} x_\vertex$.  In the
following, $\Energy$ denotes the \emphi{energy} of the \LP solution;
that is $\Energy =\EnergyX{\VSet} = \sum_{v \in \VSet} x_v$. Note that
the energy is at most the number of vertices of the hypergraph.  Also,
we assume that $\Energy\geq 1$ (which is always true since all the
capacities are at least one).

\begin{definition}
    Let $\hgraph = (\VSet, \HESet)$ be a hypergraph. For any integer
    $k$, let $\feq{k}{\cdot}$ denote the function %
    \InFullVer{%
       \begin{align*}
           \feq{k}{t} = \max_{X \subseteq \VSet, |X| \leq t}
           \cardin{\brc{\hedge\sep{ \hedge \in \HESet \text{ and }
                    \cardin{X \cap \hedge} = k+1}}};
       \end{align*}%
    }%
    \InConfVer{%
       $ \feq{k}{t} = \max_{X \subseteq \VSet, |X| \leq t}
       \cardin{\brc{\hedge\sep{ \hedge \in \HESet \text{ and }
                \cardin{X \cap \hedge} = k+1}}}$; }
    that is, $\feq{k}{t}$ is the maximum number of hyperedges of size
    $k+1$ of a sub-hypergraph of $\hgraph$ that is induced by a subset
    of at most $t$ vertices. We say that $\hgraph$ has \emphi{bounded
       growth} if the following conditions are satisfied:
    \begin{compactenum}[~~~(A)]
        \item There exists a non-decreasing function $\gamma(\cdotp)$
        such that $\feq{k}{t} \leq 2^{O(k)} t \gamma(t) $ for any $k$
        and $t$.
        \item There exists a constant $c$ such that $\feq{k}{x t} \leq
        c \feq{k}{t}$ for any $t, k$ and $x$ such that $1 \leq x \leq
        2$.
    \end{compactenum}
    
    \deflab{bounded}%
\end{definition}

\noindent
This notion of bounded growth is a hereditary property of the
hypergraph, and it is somewhat similar to the bounds on the size of
set systems with bounded \VC dimension.  Hypergraphs with bounded
growth arise naturally in geometric settings.

The minimum capacity of a packing instance is a useful measure of how
hard the instance is; formally, the \emphi{minimum capacity} of a
given instance $\hgraph$ is
\begin{equation}
    \mc = \mc(\hgraph) = \min_{\hedge \in \HESet} \capacityX{\hedge}.
    \eqlab{m:c}
\end{equation}



\InFullVer{\subsection{Basic tools}}

Let $\hgraph = (\VSet,\HESet)$ be a hypergraph, and let $\VSetA
\subseteq \VSet$ be a subset of its vertices. The sub-hypergraph of
$\hgraph$ \emphi{induced} by $\VSetA$ is $\InducedX{\hgraph}{\VSetA} =
\pth{\VSetA, \brc{ \hedge \cap \VSetA \sep{ \hedge \in \HESet}}}$.

The following lemma testifies that the packing problem can be solved
in a straightforward fashion if all the capacities are the same (i.e.,
uniform capacities). This is done by repeatedly applying a procedure  
to find and remove an independent set in the remaining induced 
sub-hypergraph (for example the procedure in \cite{ch-aamis-11}).

\InConfVer{ For some basic tools about these settings, see
   \apndref{basic:tools}.}

\newcommand{\BasicTools}{%
   \begin{lemma}
       Let $\hgraph$ be a hypergraph for which there is a polynomial
       time algorithm \Alg that takes as input a fractional solution
       to \HyperLP for an \PackHGraph instance on $\hgraph$, or any
       induced subgraph of $\hgraph$, with unit capacities --- i.e.,
       an independent set instance --- and it constructs an integral
       solution whose value is at least an $\alpha$ fraction of the
       value of the fractional solution.  Then one can compute in
       polynomial time a $2\alpha$-approximation for any instance of
       \PackHGraph on $\hgraph$ with uniform capacities (i.e., all the
       hyperedges have the same capacity, say $k$).
       
       \lemlab{basic}
   \end{lemma}
   
   \begin{proof}
       Let $\hgraph=(\VSet, \HESet)$ be an instance of \PackHGraph in
       which all hyperedges have the same capacity, say $k$. Let
       $\hgraph_0 = \hgraph$, and in the $i$\th iteration, for
       $i=1,\ldots, k$, compute a maximum weight independent set
       $\VSetB_i$ in $\hgraph_{i-1}$ using \Alg, and let $\hgraph_{i}
       = \InducedX{\hgraph}{\VSetA \setminus \VSetD_i }$, where
       $\VSetD_i = \VSetB_1\cup \ldots \cup \VSetB_{i}$. We claim that
       $\VSetD_k$ is the required approximation.
       
       Clearly, no hyperedge of $\hgraph$ contains more than $k$
       vertices of $\VSetD_k$ as it is the union of $k$ independent
       sets, and as such it is a valid solution. Now, let $\VOpt$ be
       the optimal solution. If $w( \VOpt \cap \VSetD_k) \geq
       w(\VOpt)/2$ then we are done. Otherwise, consider the
       hypergraph $\hgraph_{i-1}$, and observe that $\VOpt \setminus
       \VSetD_{i-1}$ is a valid solution for \PackHGraph for this
       graph (with uniform capacities $k$). Interpreting this integral
       solution as a solution to the \LP, and scaling it down by $k$,
       we get a factional solution to the independent set \LP of this
       hypergraph of value $w(\VOpt \setminus \VSetD_{i-1})/k$.  Since
       $\VSetB_i$ was constructed using \Alg on the optimal fractional
       solution to the independent set \LP of this hypergraph, we have
       that
       \[
       w(\VSetB_i) \geq \OptLP(\hgraph_{i-1})/\alpha \geq w( \VOpt
       \setminus \VSetD_{i-1})/k\alpha \geq w( \VOpt \setminus
       \VSetD_{k})/k\alpha \geq w(\VOpt)/2k\alpha.
       \]
       Which implies that $w(\VSetD_k) \geq w(\VOpt)/2\alpha$.
   \end{proof}
   
   A hypergraph $\hgraph = (\VSet,\HESet)$ \emphi{shatters} $\VSetA
   \subseteq \VSet$ if the number of hyperedges in
   $\InducedX{\hgraph}{\VSetA}$ is $2^{\cardin{\VSetA}}$.  The
   \emphi{\VC dimension} of $\hgraph$ is the size of the largest set
   of vertices it shatters.
   
   The following is a ``sparsification'' lemma. Here we get better
   bounds than the standard technique, as we are using stronger
   sampling results known for spaces with bounded \VC dimension.
   
   \begin{lemma}
       Let $\hgraph = (\VSet, \HESet)$ be an instance of \PackHGraph
       with \VC dimension $d$, and consider its fractional \LP
       solution of value $\Opt$ and with energy $\Energy$. Then, one
       can compute, in polynomial time, a valid fractional solution
       for the \LP of $\hgraph$ such that:
       \begin{compactenum}[(A)]
           \item The value of the new fractional solution is $\geq
           \Opt/12$.
           \item The number of vertices with non-zero value is $O( d
           \Energy \log \Energy)$.
           
           \item The value of each non-zero variable is equal to $i/M$
           for some integer $i \leq M$, where $M = O( d \log \Energy
           )$.
           
           \item The total energy in the new solution is $\Theta(
           \Energy)$.
       \end{compactenum}
       
       \lemlab{sparsify}
   \end{lemma}
   
   \begin{proof}
       Let $\eps = 1/\Energy$, where $\Energy =\sum_\vertex
       x_\vertex$, and $x_\vertex$ is the value the \LP assigns to
       $\vertex \in \VSet$ in the optimal \LP solution.
       Let $T = O\pth{ d \log \Energy }$ and $\RSample$ be a random
       sample of $\VSet$ of (expected) size $\tau = \Energy T = O(
       (d/\eps) \log(1/\eps))$, created by picking each vertex
       $\vertex$ independently with probability $ x_\vertex \cdot T$%
       \footnote{%
          A minor technicality is that $x_\vertex T$ might be larger
          than one. In this case, we put $\floor{ x_\vertex T}$ copies
          of $\vertex$ into $\RSample$, and we put an extra copy of
          $\vertex$ into $\RSample$ with probability $x_\vertex T -
          \floor{x_\vertex T}$. It is straightforward to verify that
          our argumentation goes through in this case. Observe that
          such large values work in our favor by decreasing the
          probability of failure.
       }. %
       This sample is a relative $(\eps, 1/2)$-approximation
       \cite{h-gaa-11, hs-rag-11}, with probability of failure $\leq
       \rho_1 = \eps^{O(d)}$. That implies that for any hyperedge
       $\hedge \in \HESet$ such that $x(\hedge) = \sum_{\vertex \in
          \hedge} x_\vertex$ we have
       \[
       \frac{ \cardin{\RSample \cap \hedge} }{ \cardin{\RSample} }%
       \leq%
       (1+1/2)\pth{ \frac{x(\hedge)}{\Energy} + \eps }.
       \]
       To interpret the above, observe that $\Ex{\cardin{\RSample \cap
             \hedge} \MakeSBig\!\ts } = x(\hedge) T$ and
       $\Ex{\cardin{\RSample} \MakeSBig\!\ts} = \Energy T$, as such, a
       rough estimate of the expectation of $\ds \cardin{\RSample \cap
          \hedge} / \cardin{\RSample} $ is ${x(\hedge)} /
       {\Energy}$. Thus, the above states (somewhat opaquely) that no
       hyperedge is being over-sampled by $\RSample$.
       
       Since the expected size of $\RSample$ is $\tau$, by Chernoff's
       inequality, we know that $\cardin{\RSample} \leq 2 \tau$ with
       probability at least $1- \rho_2$, where $\rho_2 =
       \eps^{O(d\Energy)}$ (as $\Energy\geq 1$).  Now, consider a
       hyperedge $\hedge$ with capacity $k$, and observe that
       $x(\hedge) \leq k$. As such, $\cardin{\RSample \cap \hedge}
       \leq (1+1/2)(x(\hedge)/\Energy + \eps) \cardin{\RSample} \leq
       (3/2) (k+1)\eps 2\tau \leq 6k T$. In particular, if $\vertex$
       appears $t_\vertex$ times in $\RSample$ ($\RSample$ is a
       multiset), then we assign it the fractional value $y_\vertex =
       t_\vertex/6 T$. We then have that $y(\hedge) = \sum_{\vertex
          \in \hedge} y_\vertex \leq \cardin{\RSample \cap \hedge} / 6
       T \leq k$ (and this holds for all hyperedges with probability
       $\geq 1- \rho_1$). As such, the fractional solution defined by
       the $y$'s is valid.
       
       As for the value of this fractional solution, consider the
       random variable $Z = \sum_{\vertex} y_\vertex w(\vertex)$,
       which is a function of the random sample $\RSample$. Observe
       that
       \[
       \Ex{Z \MakeBig \ts }
       =%
       \sum_{\vertex} w(\vertex) \Ex{ \frac{t_\vertex}{6 T} \MakeBig
       }%
       =%
       \sum_{\vertex} \frac{ w(\vertex)}{6 T} \Ex{ t_\vertex \MakeBig
       }%
       =%
       \sum_{\vertex} \frac{ w(\vertex)}{6 T} x_\vertex T
       =%
       \frac{1}{6}\sum_{\vertex} w(\vertex) x_\vertex%
       =%
       \frac{\OptLP}{6},
       \]
       as $\OptLP = \sum_{\vertex} w(\vertex) x_\vertex$.  In
       particular, since no vertex can have $w(\vertex) > \OptLP$
       (otherwise, we would choose it as the solution), it follows
       that
       \[
       Z = \sum_{\vertex} y_\vertex w(\vertex)%
       \leq%
       \OptLP \sum_{\vertex} \frac{t_\vertex}{ 6 T}%
       \leq%
       \OptLP \, \frac{\cardin{\RSample}}{ 6 T}%
       \leq%
       \OptLP \frac{2\tau}{6T}%
       = \OptLP \frac{2\Energy T}{6T}%
       \leq%
       \frac{\Energy}{3}\OptLP.
       \]
       This implies that $\Prob{ Z \geq \OptLP/12 \MakeSBig\ts } \geq
       1/4\Energy = \eps/4$. Indeed, if not,
       \begin{align*}
           \Ex{Z \MakeSBig\ts}%
           &\leq%
           \frac{\OptLP}{12}\Prob{Z \leq \frac{\OptLP}{12}}%
           +%
           \Prob{\MakeBig Z \geq \frac{\OptLP}{12} } \frac{\Energy}{3}
           \OptLP%
           \\
           &<%
           \frac{\OptLP}{12} + \frac{1}{4\Energy} \cdot
           \frac{\Energy}{3} \OptLP %
           \leq%
           \frac{\OptLP}{6},
       \end{align*}
       a contradiction. As such, a random sample $\RSample$
       corresponds to a valid assignment with value at least
       $\OptLP/12$ with probability at least $\Prob{ Z \geq \OptLP/12
       } - \rho_1 - \rho_2 \geq \eps/8$, as $\rho_1 + \rho_2$ is an
       upper bound on the sample $\RSample$ failing to have the
       desired properties. As such, taking $u =O( \Energy \log
       \Energy)$ independent random samples one of them is the
       required assignment, with probability $\geq 1 - (1-\eps/8)^u
       \geq 1 - 1/\Energy^{O(1)}$. We take this good sample together
       with its associated \LP values as the desired fractional
       solution to the \LP. Also, note that the total energy of the
       new solution is $\Theta(\Energy)$, since by Chernoff's
       inequality $\tau/2 \leq \cardin{\RSample}\leq 2\tau$ with
       probability at least $1-2\rho_2$.
   \end{proof}
}   

\InFullVer{%
   \BasicTools%
}


\section{Approximate packing for hypergraphs}
\seclab{hypergraph}

In this section, we present the algorithm for computing a packing for
a given hypergraph $\hgraph = (\VSet, \HESet)$. We assume that
$\cardin{\HESet}$ is polynomial in $\cardin{\VSet}$ and that $\hgraph$
has the bounded growth property introduced in \defref{bounded}
(properties which both hold for hypergraphs arising out natural
geometric settings). Let $x$ be a solution to the \HyperLP relaxation
described in \secref{relax}.

\subsection{The algorithm}
\seclab{algorithm}
We round the fractional solution to an integral solution using a
standard randomized rounding with alteration approach. The first step
is to choose an appropriate ordering of the vertices. We will see
later how to choose a good ordering; for now, we assume that we are
given the ordering. The rounding then proceeds in two phases, the
\emphi{selection phase} and the \emphi{alteration phase}. In the
selection phase, we pick a random sample $\RSet$ of the vertices by
picking each vertex $\vertex$ independently at random with
probability $x_{\vertex} / \scl$, where $\scl$ is a parameter that we
will determine later. In the alteration phase, we pick a subset of
$\RSet$ as follows: We consider the sampled vertices in the order
chosen and we add the current vertex to our solution if the resulting
solution remains feasible. We say that a vertex is \emphi{selected}
if it is present in the sample, and we say that it is
\emphi{accepted} if it is present in the solution. The main insight
is that we can take advantage of the bounded growth property of the
hypergraph to show that there is an ordering such that each vertex
is accepted with constant probability, provided that it is selected.
This will immediately imply that the algorithm achieves a
$O(\scl)$-approximation.

\medskip

The main challenge is to prove that a good ordering for the alteration phase exists, that is an ordering such that we accept each selected vertex with constant probability.  We now proceed to give such a proof.  This proof will suggest a natural $O(n^{C+O(1)})$ time brute force algorithm to actually compute this good ordering, where $C$ is the maximum capacity of an edge in the given instance.  In \secref{ordering:polytime} we show how one can remove this exponential dependence on $C$.

\begin{rexample}
    To keep the presentation accessible, we interpret this algorithm
    for instances of \PackRegions in which the regions are disks. 
    Specifically, we are given a weighted set of disks
    $\ObjSet$ and set of points $\PntSet$ with capacities. The disks
    of $\ObjSet$ form the set of vertices of the hypergraph and every
    point $\pnt \in \PntSet$ forms a hyperedge; that is the hyperedge
    $\hedge_\pnt$ is the set of all disks of $\ObjSet$ that contain
    $\pnt$.
    
    In this case, the mysterious quantity $\feq{k}{t}$ (see
    \defref{bounded}) is bounded by the number of faces in an
    arrangement of $t$ disks that are of depth exactly $k+1$.  Since
    the union complexity of $t$ disks is linear, standard application
    of the Clarkson technique implies that $\feq{k}{t} = O(kt)$. In
    particular, for this case $\gamma(t) =O(1)$.
\end{rexample}

\subsubsection{Constructing a good ordering}
Before we describe how to construct a good ordering of the vertices,
it is useful to understand what will force a vertex to be rejected in
the alteration phase. With this goal in mind, consider an ordering of
the vertices. Let $\RSet$ be a sample of the vertices in $\VSet$ such
that each vertex $\vertex$ is in $\RSet$ independently at random with
probability $x_{\vertex} / \scl$. Let $\vertex$ be a vertex in
$\RSet$. When we consider $\vertex$ in the alteration phase, we will
reject $\vertex$ iff there exists a hyperedge $\hedge$ of capacity
$\capacityX{\hedge}$ such that $\hedge$ contains $\vertex$ and we have
already accepted $\capacityX{\hedge}$ vertices of $\hedge$.  The event
that we already accepted $\capacityX{\hedge}$ vertices of $\hedge$ is
difficult to analyze.  However, as we will see, we can settle for a
more conservative analysis that upper bounds the probability that
$\vertex$ is rejected, given that \emph{all} of the vertices in
$\RSet$ that appear before $\vertex$ in the ordering are accepted. (In
the alteration phase, it is possible that not all vertices in $\RSet$
that appear before $\vertex$ will be accepted, but this can only help
us.) Since we are only interested in the event that $\RSet$ contains
$k + 1$ vertices --- the vertex $\vertex$ and $k$ other vertices that
appear before $\vertex$ in the ordering --- that are contained in a
hyperedge of capacity $k$, only the set of vertices that appear before
$\vertex$ in the sample matter, and not the actual ordering of the
vertices. With this observation in mind, we define a
\emphi{$k$-conflict} to be a set of $k + 1$ vertices that are
contained in a hyperedge of capacity $k$.  In the following,
$\conflictSet_k$ denotes the set of all $k$-conflicts, and
$\conflictSet = \cup_k \conflictSet_k$ denotes the set of all
conflicts. We are interested in the probability of the event that all
of the vertices of a $k$-conflict, $\conflict$, are present in the 
sample, and we refer to this probability as the \emphi{$\scl$-potential} 
of the conflict, $\forceY{\scl}{\conflict}$.  For the analysis it will 
also be useful to define the unscaled version of this quantity, that 
is the probability that all the vertices of a conflict are present 
given that we sampled each vertex with probability $x_\vertex$ and not 
$x_\vertex/\scl$.  We refer to this quantity as simply the 
\emphi{potential} of the conflict, $\forceX{\conflict}$.  
Formally we have, 
\[
\forceY{\scl}{\conflict} = \prod_{\vertex \in \conflict}
\frac{x_{\vertex}}{\scl} \qquad \text{ and } \qquad \forceX{\conflict}
= \prod_{\vertex \in \conflict} x_{\vertex}.
\]
Another quantity of interest is the expected number of conflicts in
which a vertex $\vertex$ participates, given that $\vertex$ is in the
sample. We refer to this quantity as the \emphi{$\scl$-resistance} of
a vertex $\vertex$ in a set of vertices $\VSetA \subseteq \VSet$, and
we use $\resistZ{\scl}{\vertex}{ \VSetA}$ to denote it: %
\InFullVer{%
   \[
   \resistZ{\scl}{\vertex}{ \VSetA} %
   =%
   \frac{\scl}{x_{\vertex} } \sum_{\conflict \in \conflictSet,
      \conflict \subseteq \VSetA, \vertex \in \conflict}
   \forceY{\scl}{\conflict}.
   \]%
}%
\InConfVer{%
   $\ds \resistZ{\scl}{\vertex}{ \VSetA} %
   =%
   \frac{\scl}{x_{\vertex} } \sum_{\conflict \in \conflictSet,
      \conflict \subseteq \VSetA, \vertex \in \conflict}
   \forceY{\scl}{\conflict}$.
}%

\paragraph{The ordering.}
Note that, if the $\scl$-resistance of $\vertex$ with respect to the
set $X$ of vertices that come before it in the ordering is small, the
probability of rejecting $\vertex$ is also small. This suggests that
the vertex with least resistance (with respect to $\VSet$) should be
the last vertex in the ordering. This gives us the following algorithm
for constructing an ordering: We compute the vertex of least
resistance and put it last in our ordering (i.e., it is
$\vertex_n$). We then recursively consider the remaining vertices and
we compute an ordering for them. In the following, we assume for
simplicity that the resulting ordering is $\vertex_1,\ldots,
\vertex_n$.

Note that computing the resistance of a vertex by brute force takes
$O\pth{ n^{C+O(1)} }$ time, where $C$ is the maximum capacity of a
hyperedge, and therefore this algorithm is not efficient. We give a
polynomial time algorithm for constructing the ordering in
\secref{ordering:polytime}.


\subsection{Analysis}

Our main insight is that, if the hypergraph satisfies the bounded
growth property defined in \defref{bounded}, then for any set $\VSetA
\subseteq \VSet$ there exists a vertex $\vertex \in \VSetA$ such that
$\resistZ{\scl}{\vertex}{ \VSetA} \leq 1/4$.  We prove this below in
\secref{proof:resistance} (see \lemref{resistance}).  This proof 
requires that we set $\scl= \sclC\gamma(\Energy)^{1/\mc}$, where 
$\sclC$ is some sufficiently large constant.  As such, in the 
remainder of this section we assume $\scl= \sclC\gamma(\Energy)^{1/\mc}$.

We now show that given $\resistZ{\scl}{\vertex}{ \VSetA} \leq 1/4$, 
proving the quality of approximation of the algorithm is
straightforward.


\begin{lemma}
    Let $\RSet$ and $\OSet$ be the set of vertices that were selected and
    accepted by the algorithm, respectively. For each $i$, we have
    $\Prob{\vertex_i \in \OSet \sep{\vertex_i \in \RSet}} \geq 3/4$.
    
    \lemlab{prob}
\end{lemma}
\begin{proof:in:appendix}{\lemref{prob}}
	Let $\VSetA_i = \permut{\vertex_1, \ldots,\vertex_i}$. Note that,
	if we selected $\vertex_i$, we rejected $\vertex_i$ in the
	alteration phase only if $\vertex_i$ participates in a conflict
	with some of the vertices in $\{\vertex_1, \ldots, \vertex_{i -
	1}\} \cap \RSet$. Let $Z_i$ be the number of conflicts of
	$\VSetA_i$ that contain $\vertex_i$ and are realized in $\RSet$,
	i.e., $\conflict \subseteq \RSet$. In the following, we show that
	the probability that $Z_i$ is non-zero is at most $1/4$, which
	implies the lemma.
    
	Consider a $k$-conflict $\conflict = \brc{ \vertex_{j_1}, \ldots,
	\vertex_{j_k}, \vertex_i}$, where each vertex of $\conflict$ is
	in $\VSetA_i$ and $\conflict$ contains $\vertex_i$.  The
	probability that all of the vertices of $\conflict$ are selected,
	given that $\vertex_i$ is selected, is equal to
    $\ds \frac{x_{j_1}}{\scl} \cdot \frac{x_{j_2}}{\scl} \cdots
    \frac{x_{j_k}}{\scl} = \frac{\scl}{x_i} \forceY{\scl}{\conflict}$.
    Therefore we have
    \[
    \Ex{Z_i \sep{ \vertex_i \in \RSet} } =%
    \sum_{\conflict \in \conflictSet, \conflict \subseteq \VSetA_i,
       \vertex_i \in \conflict} \frac{\scl}{x_i} \forceY{\scl}{\conflict}%
    =%
    \resistZ{\scl}{\vertex_i}{ \VSetA_i} \leq \frac{1}{4},
    \]
    where the last inequality follows from \lemref{resistance}
    and the fact that $\vertex_i$ is the vertex of minimum resistance in
    $\VSetA_i$. Thus
    \[
    \Prob{\vertex_i \notin \OSet \sep{\vertex_i \in \RSet}}%
    \leq%
    \Prob{Z_i > 0 \sep{\vertex_i \in \RSet}} \leq%
    \Ex{Z_i \sep{ \vertex_i \in \RSet} }%
    \leq%
    \frac{1}{4}.
    \]
    Therefore, if $\vertex_i$ is selected, it is accepted with
    probability at least $3/4$.
\end{proof:in:appendix}

\begin{corollary}
	The total expected weight of the set of vertices output by the
	algorithm is $\Omega \pth{\Opt /\gamma(\Energy)^{1/\mc}}$, where
	$\Opt$ is the weight of the optimal solution, and $\mc$ is the
	minimum capacity of the given instance.
    
    \corlab{result}
\end{corollary}

\begin{proof:in:appendix}{\corref{result}}
    By \lemref{prob}, for each vertex $\vertex \in \VSet$, we have
    \begin{align*}
        \Prob{ \mbb \vertex \in \OSet}%
        &=%
        \Prob{ \mbb \MakeSBig \pth{\vertex \in \OSet} \cap \pth{
		\vertex
              \in \RSet}}%
        =%
        \Prob{\vertex \in \OSet \sep{\vertex \in \RSet}} \Prob{ \mbb
		\vertex
           \in \RSet}%
        \geq%
        \frac{3}{4} \Prob{ \MakeSBig \vertex \in \RSet}\\
        &\geq%
        \frac{3x_{\vertex}}{4 \scl},
    \end{align*}
	where $\scl = O \pth{ \gamma(\Energy)^{1/\mc} }$.  By linearity
	of expectation, we have that the expected weight of the generated
	solution is at least
    \[
    \sum_{\vertex \in \VSet} \frac{3x_{\vertex}}{4 \scl} w_{\vertex}%
    =%
    \Omega\pth{ \frac{\sum_{\vertex} x_{\vertex}
          w_{\vertex}}{\gamma(\Energy)^{1/\mc}}}%
    =%
    \Omega\pth{ \frac{\Opt}{\gamma(\Energy)^{1/\mc}}},
    \]
    as $\sum_{\vertex} x_{\vertex} w_{\vertex}$ is the value of the
    fractional \LP solution, which is bigger than (or equal to) the
    weight of the optimal solution.
\end{proof:in:appendix}


\subsubsection{On the expected number of conflicts being realized}

To analyze the algorithm we need to understand how conflicts might
form during its execution, and show that the damage of such conflicts
to the generated solution is limited. To this end, consider the
quantity
\begin{align*}
    \feq{k}{t} = \max_{A \subseteq \VSetA, |A| \leq t}
    \cardin{\brc{\hedge\sep{ \hedge \in \HESet \text{ and } \cardin{A
                \cap \hedge} = k+1}}}.
\end{align*}
This is the maximum number of $k$-conflicts that can be realized for a
set of $t$ vertices.  The quantity of interest in the following is $
\sum_{\conflict \in \conflictSet_k} \forceX{\conflict}$, as it is the
expected number of conflicts that would be realized if we sample
according to the \LP solution. Our purpose is to prove that this
quantity is bounded by a function of the energy of the \LP (the bound
would involve the function $\feq{k}{\cdot}$ defined above).


With this goal in mind, we let $\RSample$ be a random sample of
$\VSetA$ such that each vertex $\vertex \in \VSetA$ is in $\RSample$
independently at random with probability $x_{\vertex} / 2$. We stress
that $\RSample$ is a random sample that we use for the purposes of
defining a quantity $\xkappa$ (i.e., the expected number of conflicts
realized in $\RSample$), and it should not be confused with the random
sample $\RSet$ that is used by the algorithm. In the following, we
bound $\xkappa$ from above in \lemref{upper-bound} and from below in
\lemref{lower-bound}. Putting these two bounds together imply the
desired bound on $\sum_{\conflict \in \conflictSet_k}
\forceX{\conflict}$.

A conflict $\conflict \in \conflictSet$ is \emphi{realized} in
$\RSample$ if there is a hyperedge $\hedge \in \HESet$ such that $
\conflict = \hedge \cap \RSample$ and $\cardin{\conflict} =
\capacityX{\hedge} + 1$.  

The following is similar in spirit to the Clarkson technique (a
similar but simpler argument was used by Chan and Har-Peled
\cite{ch-aamis-11}).

\begin{lemma}
    The expected number of $k$-conflicts realized in $\RSample$ is
    $\xkappa = O(\feq{k}{\Energy(X)})$, where $\RSample$ is a random
    sample of $\VSetA$ such that each vertex $\vertex \in \VSetA$ is
    in $\RSample$ independently at random with probability
    $x_{\vertex} / 2$.
    
    \lemlab{upper-bound}
\end{lemma}
\begin{proof:in:appendix}{\lemref{upper-bound}}
    Each $k$-conflict $\conflict$ that is realized corresponds to a
    hyperedge $\hedge$ with capacity $k$ such that $\conflict = \hedge
    \cap \RSample$. Additionally, two realized conflicts that are
    distinct correspond to different hyperedges.  Therefore the number
    of $k$-conflicts that are realized in $\RSample$ is at most the
    number of hyperedges $\hedge$ such that the capacity of $\hedge$
    is $k$ and $|\hedge \cap \RSample| = k + 1$.  It follows from the
    definition of $\feq{k}{\cdotp}$ that the number of $k$-conflicts
    is at most $\feq{k}{|\RSample|}$. Therefore it suffices to upper
    bound the expected value of $\feq{k}{|\RSample|}$.
    
    Note that $\Ex{\cardin{\RSample}} = \Energy(\VSetA)/2$. We have
    \begin{align*}
        \Ex{ \feq{k}{|\RSample|} \MakeBig\! }%
        &\leq%
        \sum_{t=0}^\infty \Prob{\cardin{\RSample} \geq t
           \frac{\Energy(\VSetA)}{2}} \feq{k}{ \MakeSBig
           (t+1)\frac{\Energy(\VSetA)}{2}}%
        \leq%
        \sum_{t=0}^\infty 2^{-(t+1)/2} \feq{k}{ \MakeSBig
           (t+1)\frac{\Energy(\VSetA)}{2}}%
        \\
        &\leq%
        \sum_{t=0}^\infty 2^{-(t+1)/2} c^{O( \log t)} \, \feq{k}{
           \MakeSBig \Energy(\VSetA)}%
        =%
        O\pth{ \MakeBig \feq{k}{ \Energy(\VSetA) } },
    \end{align*}
    since $\hgraph$ has the bounded growth property (see
    \defref{bounded}), and by the Chernoff inequality (we use here
    implicitly that $\Energy(\VSetA) \geq 1$).
\end{proof:in:appendix}

\begin{lemma}
    For each $k$-conflict $\conflict$, the probability that
    $\conflict$ is realized in $\RSample$ is at least
    $\forceX{\conflict} / 2(2e)^{k}$. Therefore the expected number of
    $k$-conflicts realized in $\RSample$ is $\xkappa =
    \Omega\left(\left(\sum_{\conflict \in \conflictSet_k, \conflict
               \subseteq X} \forceX{\conflict}\right) /
        (2e)^k\right)$.
    
    \lemlab{lower-bound}
\end{lemma}
\begin{proof:in:appendix}{\lemref{lower-bound}}
    Let $\hedge \in \HESet$ be a hyperedge with capacity $k$ that
    generated the conflict $\conflict$. Since $x$ is a feasible
    solution for the \LP, we have that $\sum_{\vertex \in \hedge -
       \conflict} x_{\vertex} \leq \sum_{\vertex \in \hedge}
    x_{\vertex} \leq \capacityX{\hedge} = k$.  Clearly, the conflict
    $\conflict$ is realized if we pick all the vertices of
    $\conflict$, and none of the vertices of $\hedge - \conflict$, and
    the probability of that event is
    \begin{align*}
        \prod_{\vertex \in \conflict} \frac{x_{\vertex}}{2}
        \prod_{\vertex \in \hedge - \conflict} \pth{1 -
           \frac{x_{\vertex}}{2}} &\geq%
        \frac{1}{2^{k+1}} \prod_{\vertex \in \conflict} x_{\vertex}
        \prod_{\vertex \in \hedge - \conflict} \exp \pth{ -
           x_{\vertex}}%
        \\
        &=%
        \frac{\forceX{\conflict}}{2^{k+1}} \cdot \exp \pth{ -
           \sum_{\vertex \in \hedge - \conflict} x_{\vertex}}%
        \geq%
        \frac{\forceX{\conflict}}{2(2e)^k},
    \end{align*}
    In the first line we used the inequality $1 - x_{\vertex} / 2 \geq
    \exp(-x_{\vertex})$, which holds since $x_{\vertex} \leq 1$.
\end{proof:in:appendix}

Putting the above two lemmas together, we get the following.

\begin{lemma}
    For any non-negative integer $k$ we have $\ds\sum_{h \in
       \conflictSet_k, h \subseteq \VSetA} \forceX{\conflict} =
    O\pth{(2e)^{k} \feq{k}{\Energy(\VSetA)} \MakeBig\! }$.
    
    \lemlab{total-energy}
\end{lemma}

\begin{rexample}
    In our running example, we have that the expected number of
    $k$-conflicts that are being realized by a random sample (sampling
    more or less according to the \LP values) is $ \sum_{\conflict \in
       \conflictSet_k} \forceX{\conflict} = O\pth{(2e)^{k} k\Energy
       \MakeBig\! }$. This is a hefty quantity,
    but the key observation is that if we sample according to the \LP
    values scaled down by a large enough constant, then the
    probability of such a conflict to be realized drops exponentially
    with $k$. In particular, for a sufficiently large constant, the
    expected number of realized $k$-conflicts in such a smple is going
    to be $\leq \Energy/\pth[]{10 \cdot 2^k}$. Intuitively, this
    implies that such conflicts can only cause the algorithm to drop
    very few vertices during the rounding stage, thus guaranteeing a
    good solution.
\end{rexample}

\subsubsection{Resistance is futile, if you pick the right vertex}

\seclab{proof:resistance}

In the following, we consider a subset $\VSetA$ of the vertices and we
show that there exists a vertex $\vertex \in \VSetA$ whose
$\scl$-resistance $\resistZ{\scl}{\vertex}{\VSetA}$ is at most $1/4$.
Recall that $\conflictSet_k$ is the set of all $k$-conflicts involving
vertices in $\VSet$. We can rewrite the $\scl$-resistance of $\vertex$
in $\VSetA$ as
\begin{align*}
    \resistZ{\scl}{\vertex}{\VSetA} =
    \frac{\scl}{x_{\vertex} } \sum_{\conflict \in \conflictSet, \conflict
       \subseteq \VSetA, \vertex \in \conflict}
       \forceY{\scl}{\conflict} =
    \frac{1}{x_{\vertex}} \sum_{k =
       \mc}^{\infty} \frac{1}{\scl^k} \sum_{\conflict \in
       \conflictSet_k, h \subseteq \VSetA, \vertex \in \conflict}
    \forceX{\conflict}.
\end{align*}
As shown in \lemref{total-energy}, we can relate the total potential
of the conflicts of $\conflictSet_k$ that are contained in $\VSetA$ to
the maximum number of $k$-conflicts contained in a set of at most
$\Energy(\VSetA)$ vertices, where $\Energy(\VSetA) = \sum_{\vertex \in
   \VSetA} x_{\vertex}$.

Recall that the hypergraph has the bounded growth
property (see \defref{bounded}) and this property is
hereditary. Therefore the function $\feq{k}{\cdotp}$ in the lemma
above has the two properties described in \defref{bounded} and we get
the following corollary.

\begin{corollary}
    We have $\sum_{h \in \conflictSet_k, h \subseteq \VSetA}
    \forceX{\conflict} = O\pth{2^{ck} \Energy(\VSetA)
       \gamma(\Energy(\VSetA)})$, where $c$ is a constant.
    
    \corlab{total-energy}
\end{corollary}

We can use \corref{total-energy} to complete the
proof of \lemref{resistance} as follows.

\begin{lemma}
    Suppose that the hypergraph $\hgraph$ satisfies the bounded growth
    property (see \defref{bounded}). Let $\scl = \sclC\,
    \gamma(\Energy)^{1/\mc}$, where $\sclC > 0$ is a sufficiently
    large constant and $\mc$ is the minimum capacity of the given
    instance (see \Eqref{m:c}).  Then, for any set $\VSetA \subseteq
    \VSet$, there exists a vertex $\vertex \in \VSetA$ such that
    $\resistZ{\scl}{\vertex}{ \VSetA} \leq 1/4$.
    
    \lemlab{resistance}
\end{lemma}

\begin{proof}
    Let $T = \sum_{\vertex \in \VSetA} x_{\vertex}
    \resistZ{\scl}{\vertex}{\VSetA}$. The quantity $T /
    \Energy(\VSetA)$ is the weighted average of the resistances of the
    vertices in $\VSetA$, where the weight of a vertex $\vertex$ is
    $x_{\vertex} / \Energy(\VSetA)$.  Therefore it suffices to show
    that $T \leq \Energy(\VSetA) / 4$, since the minimum resistance is
    at most the weighted average. We have
    \begin{align*}
        T &=%
        \sum_{k=\mc}^\infty \frac{1}{\scl^k} \sum_{\vertex \in \VSetA}
        \sum_{\substack{\conflict \in \conflictSet_k\\
              \conflict \subseteq
              \VSetA\\
              \vertex \in \conflict}} \forceX{\conflict} =%
        \sum_{k=\mc}^\infty \frac{k + 1}{\scl^k}
        \sum_{\substack{\conflict \in \conflictSet_k\\
              \conflict \subseteq \VSetA }} \forceX{\conflict} =%
        \sum_{k=\mc}^\infty
        \frac{k+1}{\scl^k} \sum_{\substack{\conflict \in \conflictSet_k\\
              \conflict \subseteq \VSetA}} \forceX{\conflict}%
        \\ &=%
        \sum_{k=\mc}^\infty \frac{k+1}{\scl^k} O\pth{ 2^{ck}
           \Energy(\VSetA) \gamma(\Energy(\VSetA)) \MakeBig\!}
        %
        \leq %
        \Energy(\VSetA) \cdot \underbrace{\beta \sum_{k=\mc}^\infty {
              \pth{\frac{2^c}{\scl}}^k {} (k+1)
              \gamma(\Energy(\VSetA)})}_{=S},
    \end{align*}
    by \corref{total-energy}, where $\beta$ is some constant.  Since
    $\scl = \sclC\gamma(\Energy)^{1/\mc}$, we have
    \begin{align*}
        S &=%
        { \sum_{k=\mc}^\infty \beta \pth{\frac{2^c}{\sclC}}^k {} (k+1)
           \pth{ \frac{1}{\gamma(\Energy)}}^{k/\mc}
           \gamma(\Energy(\VSetA))}%
        \leq%
        { \sum_{k=\mc}^\infty \beta \pth{\frac{2^c}{\sclC}}^k {} (k+1)
           \pth{ \frac{1}{\gamma(\Energy)}}^{k/\mc} \gamma(\Energy)}%
        \leq%
        \frac{1}{4},
    \end{align*}
    In the second to last inequality, we have used the fact that
    $\gamma(\cdotp)$ is non-decreasing. The last inequality follows if
    we pick $\sclC$ to be a sufficiently large constant.  Therefore $T
    \leq \Energy(\VSetA) / 4$, and the lemma follows.
\end{proof}


\newcommand{\ImproveRunningTime}{%
\subsection{Improving the running time}
\seclab{ordering:polytime}%
\apndlab{improve:r:t}%

In \secref{hypergraph}, we described an algorithm that
constructs an ordering of the vertices by repeatedly finding the
vertex of least resistance with respect to the set of remaining
vertices.  Computing the resistance of a vertex by brute force takes
$O(n^{C+O(1)})$ time, where $C$ is the maximum capacity of an edge in
$\HESet$.  However, for our analysis to go through, we only need to
find a vertex that is safe with respect to the set of remaining
vertices; informally, a vertex $\vertex$ is safe if the probability
that it participates in a conflict with a random
sample of the remaining vertices is smaller than some constant (that
is strictly smaller than one), where each remaining vertex $u$ is
included in the sample with probability $x_u / \scl$.  In this section
we show that there is a sampling algorithm that finds a safe vertex
with high probability and its running time is polynomial in the
maximum capacity $C$.

\begin{lemma}
	Computing a good ordering of the vertices can be done in
	polynomial time. Namely, the algorithm of \secref{hypergraph} can
	be implemented in polynomial time.
\end{lemma}

\begin{proof}
    To get the same quality of approximation we do not need to take
    the vertex of least resistance in each round (of computing the
    ordering), but merely a vertex that is ``safe.'' More precisely,
    let $\VSetA$ be the current set of vertices, let $\vertex$ be a
    vertex of this set, and let $\RSample$ be a random sample of
    $\VSetA$ in which each vertex $u$ is included with probability
    $x_u / \scl$ (also we force $\vertex$ to be in $\RSample$).  We
    say that $\vertex$ is \emphi{violated} in $\RSample$ if $\vertex$
    is contained in a hyperedge $\hedge$ such that the number of
    vertices of $\hedge$ that are in $\RSample$ is larger than its
    capacity $\capacityX{\hedge}$. Let $\mu(\vertex, \VSetA)$ denote
    the probability that $\vertex$ is violated in $\RSample$. Note
    that $\mu(\vertex, \VSetA)$ is a (conservative) upper bound on the
    probability that $\vertex$ is rejected by our rounding algorithm
    if we started with an ordering in which $\VSetA\setminus
    \{\vertex\}$ is the set of all vertices that come before
    $\vertex$. Therefore, in order for our rounding to succeed, in
    each round we only need to find a vertex $\vertex$ for which the
    probability $\mu(\vertex, \VSetA)$ is low, where $\VSetA$ is the
    set of all vertices that still need to be ordered at the beginning
    of the round. (We remark that it follows from the argument of
    \lemref{resistance} that, for any set $\VSetA$, there is a vertex
    $\vertex$ for which $\mu(\vertex, \VSetA) \leq 1/4$.)
    
    Now we are ready to describe how to construct an ordering for our
    algorithm. Let $\VSetA$ be the set of vertices that still need to
    be ordered. As we will see shortly, for each vertex $\vertex \in
    \VSetA$, we can compute an estimate $\overline{\mu}(\vertex,
    \VSetA )$ of the probability $\mu(\vertex, \VSetA )$. We pick the
    vertex $\vertex$ with minimum estimated probability
    $\overline{\mu}(\vertex, \VSetA)$, we make $\vertex$ the last
    vertex (in the ordering of $\VSetA$) and we recursively order
    $\VSetA \setminus \{\vertex\}$.
    
    We can compute the estimates $\overline{\mu}(\vertex, \VSetA)$ in
    polynomial time as follows. Fix a vertex $\vertex$. Let $\Pnum$ be
    a sufficiently large polynomial in $n$. We pick $\Pnum$
    independent random samples of $\VSetA$ (again, forcing $\vertex$
    to be in each of these samples); in each random sample, each
    vertex $u$ is included with probability $x_u / \scl$. We set
    $\overline{\mu}(\vertex, \VSetA )$ to be the fraction of the
    samples in which the vertex $\vertex$ is violated. Using a
    standard argument based on the Chernoff inequality, we can show
    that our estimates are very close with high probability, and
    therefore our rounding algorithm achieves the required
    approximation with high probability as well; we omit the easy but
    tedious details.
\end{proof}
}

\InConfVer{See \apndref{improve:r:t} for details of how to improve the
   running time.}

\InFullVer{\ImproveRunningTime}

\subsection{The result}
\begin{theorem}
    Let $\hgraph = (\VSet, \HESet)$ be a hypergraph with a weight
    function $\weight{\cdot}$ on the vertices and a capacity function
    $\capacityX{\cdot}$ on the edges, such that $\cardin{\HESet}$ is
    polynomial in $\cardin{\VSet}$ and $\hgraph$ satisfies the bounded
    growth property (see \defref{bounded}); that is, $\feq{k}{t} =
    2^{O(k)} t \gamma(t)$. Then we can compute in polynomial time a
    subset $X \subseteq \VSet$ of vertices such that no hyperedge
    $\hedge$ contains more than its capacity $\capacityX{\hedge}$
    vertices of $X$.  Furthermore, in expectation, the total weight of
    the output set is $\Omega \pth{\Opt /\gamma(\Energy)^{1/\mc}}$,
    where $\Opt$ is the weight of the optimal solution, and $\mc$ is
    the minimum capacity of the given instance.
    
    \thmlab{p:hypergraph}
\end{theorem}

\InConfVer{
\section{Conclusions}
\seclab{conclusions}%

In this paper, we presented a general framework for approximating
geometric packing problems with non-uniform constraints. We then
applied this framework in a systematic fashion to get improved
algorithms for specific instances of this problem, many of which
required additional non-trivial ideas. There are several special cases
of this problem for which we currently do not know any useful
approximation; for example, the special case of packing axis-parallel
boxes into points, in which the boxes are in four dimensions is still
wide open. Making some progress on these special cases is an
interesting direction for future work.

\section*{Acknowledgments}

The authors thank Timothy Chan, Chandra Chekuri, and Esther Ezra for
several useful discussions.
}

\InConfVer{%
\bibliographystyle{alpha}%
\bibliography{packing}

\appendix
}

\InConfVer{%
\section{Approximate packing for hypergraphs}
}

Consider an integer constant $\cRelax >0$, and observe that one can
always relax the capacity constraints of a given instance of
\PackHGraph by replacing all capacities smaller than $\cRelax$ by
$\cRelax$. \thmref{p:hypergraph} thus implies the following.

\begin{corollary}
    Given an instance of \PackHGraph, with the bounded growth
    property, one can compute in polynomial time a $\pth[]{ O
       \pth{\gamma(\Energy)^{1/\cRelax}}, \cRelax}$-approximation to
    the optimal solution.
\end{corollary}

\subsection{Contention resolution schemes}
Chekuri et al. \cite{cvz-sfmmm-11} considered a broad class of
rounding schemes, which they called \emph{contention resolution
schemes} (CR schemes). Informally, the family of all CR schemes
consists of all rounding strategies based on randomized rounding with
alteration. The precise definition of a CR scheme is the following.

Let $N$ be a finite ground set of size $n$, and let $f: 2^{N}
\rightarrow \mathbb{R}_+$ be a submodular\footnote{A function $f: 2^N
\rightarrow \mathbb{R}$ is submodular if $f(A) + f(B) \geq f(A \cap
B) + f(A \cup B)$ for any two subsets $A, B$ of $N$. Additionally,
$f$ is monotone if $f(A) \leq f(B)$ for all subsets $A, B$ such that
$A \subseteq B$.} set function over $N$. Let $\mathcal{I} \subseteq
2^N$ be a downward-closed\footnote{A family $\mathcal{I}$ of subsets
of $N$ is downward-closed if $B \in \mathcal{I}$ and $A \subseteq B$
then $A \in \mathcal{I}$.} family of subsets of $N$. The problem of
maximizing $f(X)$ subject to the constraint that $X \in \mathcal{I}$
generalizes the hypergraph packing problem: the function $f$
satisfies $f(\VSetA) = \sum_{\vertex \in \VSetA} w_{\vertex}$, and
the family $\mathcal{I}$ is the family of all subsets $\VSetA
\subseteq \VSet$ such that, for each hyperedge $\hedge$, $|\hedge
\cap \VSetA| \leq \capacityX{\hedge}$. Let $P_{\mathcal{I}} \subseteq [0,
1]^n$ be a convex relaxation\footnote{$P_{\mathcal{I}}$ is the closure of 
the set of characteristic vectors of the sets in $\mathcal{I}$ under convex
combinations.} of the constraints imposed by $\mathcal{I}$; the set
of all feasible fractional solutions to the \HyperLP relaxation
described in \secref{relax} is a convex relaxation for the family of
all feasible solutions to the hypergraph packing problem. Let $F$ be
the multilinear extension\footnote{The multilinear extension $F: [0,
1]^{|N|} \rightarrow \mathbb{R}$ of a function $f: 2^N  \rightarrow
\mathbb{R}$ is the function $F(x) = \sum_{S \subseteq N} f(S)
\prod_{i \in S} x_i \prod_{j \notin S} (1 - x_j)$} of $f$. Let $x$ be a
feasible solution to the relaxation $\{\max F(x): x \in
P_{\mathcal{I}}\}$. The definition of the multilinear extension $F$
suggests the following natural rounding strategy: given $x$, we
construct a random set $R(x)$ by picking each $i \in N$ independently
at random with probability $x_i$. The expected value of $f(R(x))$ is
equal to $F(x)$, but it is unlikely that $R(x)$ is in $\mathcal{I}$.
To address this, we want to remove some elements from $R(x)$ in order
to get a subset $I \subseteq R(x)$ such that $I \in \mathcal{I}$. We
want this step to have the property that, for each $i \in N$, the
probability that $i$ is in $I$ is at least $cx_i$, for some parameter
$c > 0$. Chekuri et al. \cite{cvz-sfmmm-11} call such a rounding
strategy a \emph{$c$-balanced CR scheme for $P_{\mathcal{I}}$}. In
certain settings it is convenient to scale the fractional solution;
the rounding strategy described above for the hypergraph packing
problem is one such example. This motivates the following more
general CR scheme.

\begin{definition}[\cite{cvz-sfmmm-11}]
	 A $(b,c)$-balanced CR scheme for $P_{\mathcal{I}}$ is a scheme
	 such that for any $x \in P_{\mathcal{I}}$, the scheme selects an
	 independent subset $I \subseteq R(bx)$ with the following
	 property: $\Pr[i \in I \;|\; i \in R(bx)] \geq c$ for every
	 element $i \in N$.  The scheme is said to be monotone if $\Pr[i
	 \in I \;|\; R(bx) = R_1] \geq \Pr[i \in I \;|\; R(bx) = R_2]$
	 whenever $i \in R_1 \subseteq R_2$. A scheme is said to be
	 strict if $\Pr[i \in I \;|\; i \in R(bx)] = c$ for every $i$.
\end{definition}

\smallskip\noindent
Chekuri et al. showed that, if $I$ is the output of a monotone $(b,
c)$-balanced CR scheme, the expected value of $I$ is at least $c
\Ex{F(bx)}$.

\begin{theorem}
	Let $f: 2^N \rightarrow \mathbb{R}_+$ be a non-negative
	submodular function and let $x$ be a point in $P_{\mathcal{I}}$,
	where $P_{\mathcal{I}}$ is a convex relaxation for $\mathcal{I}
	\subseteq 2^N$. Let $I(x) \in \mathcal{I}$ be the random output
	of a monotone $(b,c)$-balanced CR scheme on $x \in
	P_{\mathcal{I}}$. If $f$ is non-monotone, let us assume in
	addition that the CR scheme is strict. Then $\Ex{f(I(x))} \geq c
	\Ex{F(bx)}$.
\end{theorem}

\smallskip\noindent
The rounding scheme described in \secref{algorithm} is a monotone
$(\scl, 1/4)$-balanced CR scheme on $x \in P_{\mathcal{I}}$, where
$P_{\mathcal{I}}$ is the set of all feasible solutions to the
\HyperLP relaxation. Therefore \thmref{p:hypergraph} extends to the
setting in which the total weight of a set of vertices is a monotone
submodular function instead of a linear function.

\begin{corollary}
	Let $\hgraph = (\VSet, \HESet)$ be a hypergraph. Let $w:
	2^{\VSet} \rightarrow \mathbb{R}_+$ be a weight function on the
	vertices that is non-negative, monotone, and submodular.  Let
	$\capacityX{\cdot}$ be a capacity function on the edges such that
	$\cardin{\HESet}$ is polynomial in $\cardin{\VSet}$ and $\hgraph$
	satisfies the bounded growth property (see \defref{bounded});
	that is, $\feq{k}{t} = 2^{O(k)} t \gamma(t)$. Then we can compute
	in polynomial time a subset $X \subseteq \VSet$ of vertices such
	that no hyperedge $\hedge$ contains more than its capacity
	$\capacityX{\hedge}$ vertices of $X$.  Furthermore, in
	expectation, the total weight of the output set is $\Omega
	\pth{\Opt /\gamma(\Energy)^{1/\mc}}$, where $\Opt$ is the weight
	of the optimal solution, and $\mc$ is the minimum capacity of the
	given instance.

    \thmlab{p:hypergraph-submod}
\end{corollary}


\section{Applications}%
\seclab{applications}
\apndlab{applications}

Using our main result (\thmref{p:hypergraph}), we get several
approximation algorithms for the packing problems mentioned in the
introduction. We present some of these results here.

\subsection{Packing regions with low union complexity}

Let $\ObjSet$ be a set of $n$ weighted regions in the plane, and let
the maximum union complexity of $m \leq n$ objects of $\ObjSet$ be
$\Union{m} = m\ts\union{m}$.  We assume that (i) $\Union{n}/n=
\union{n}$ is a non-decreasing function, and (ii) there exists a
constant $c$, such that $\Union{ x r} \leq c \, \Union{r}$, for any
$r$ and $1 \leq x \leq 2$.  We are also given a set of points
$\PntSet$, where each point $\pnt \in \PntSet$ is assigned a positive
integer $\capacityX{\pnt}$ which is the {capacity} of $\pnt$.

We are interested in solving \PackRegions (\probref{pack:regions}) for
$\ObjSet$ and $\PntSet$.  Consider the hypergraph $\hgraph$ obtained
by creating a vertex for each region and a hyperedge for each subset
of regions containing a given point of $\PntSet$. Here, $\feq{k}{t}$
is bounded by the number of faces in the arrangement of $t$ regions of
depth exactly $k+1$.  The number of such faces can be bounded by the
union complexity by a standard application of the Clarkson technique
\cite{c-arscg-88, cs-arscg-89}.

\begin{lemma}
    Consider a set of regions in the plane such that the boundary of
    every pair intersects a constant number of times.  The number of
    faces of depth at most $k+1$ in the arrangement of any subset of
    these regions of size $t$ is $O\pth{k^2 \Union{t/k}}$.
    \lemlab{clarkson}
\end{lemma}

\noindent
Plugging this bound into \thmref{p:hypergraph} yields the following
result.

\begin{theorem}
    Let $\ObjSet$ be a set of $m$ weighted regions in the plane such
    that the union complexity of any $t$ of them is $\Union{t} =
    t\union{t}$. Let $\PntSet$ be a set of $n$ points in the plane,
    where there is a capacity $\capacityX{\pnt}$ associated with each
    point $\pnt \in \PntSet$.  There is a polynomial time algorithm
    that computes a subset $\OSet \subseteq \ObjSet$ of regions such
    that no point $\pnt \in \PntSet$ is contained in more than
    $\capacityX{\pnt}$ regions of $\OSet$.  Furthermore, in
    expectation, the total weight of the output set is $\Omega
    \pth{\Opt /\union{\Energy}^{1/\mc}}$, where $\Opt$ is the weight
    of the optimal solution, $\Energy$ is the energy of the \LP
    solution and $\mc$ is the minimum capacity of the given instance.
    
    Alternatively, for any integer constant $\cRelax$, one can get a
    $\pth[]{ O\pth[]{ \union{\Energy}^{1/\cRelax}}, \,
       \cRelax}$-approximation to the optimal solution for the given
    instance.
    
    \thmlab{p:regions:low:union}
\end{theorem}

The following results follow from the theorem above.

\begin{corollary}
    \begin{inparaenum}[(A)]
        \item \itemlab{s:triangles}%
        The union complexity of pseudo-disks and fat triangles of
        similar size is linear; that is, $\Union{t} = O(t)$.
        Therefore we get an $O(1)$-approximation for \PackRegions if
        the regions are fat triangles of similar size, disks, or
        pseudo-disks.

        \item The union complexity of fat triangles is $O( n\log^* n)$
        \cite{abes-ibucl-11}. Therefore we get an $O(\log^*
        n)$-approximation for \PackRegions if the regions are
        (arbitrary) fat triangles.
        
        \item Consider a set of regions in the plane such that any
        pair of them intersects a constant number of times (e.g., a
        set of arbitrary triangles). In this case, $\Union{t} =
        O(t^2)$ and $\union{t} = O(t)$. Therefore, for any integer
        constant $\cRelax > 0$, we get an $\pth[]{
           O\pth{\Energy^{1/\cRelax}}, \, \cRelax}$-approximation for
        instances of \PackRegions on such regions.
    \end{inparaenum}
    
    \corlab{pseudo} \corlab{fat:triangles:into:points}
\end{corollary}


\subsection{Packing halfspaces, rays and disks}
\seclab{halfRayDisk}

\begin{problem}
    \begin{inparaenum}[(A)]
        \item {\rm{\PackHalfspaces}}: Given a weighted set of
        halfspaces $\HalfspaceSet$ and a set of points $\PntSet$ with
        capacities in $\Re^3$, find a maximum weight subset $\OSet$ of
        $\HalfspaceSet$ so that, for each point $\pnt$, the number of
        halfspaces of $\OSet$ that contains $\pnt$ is at most
        $\capacityX{\pnt}$.
        
        \item {\rm{\PackRaysInPlanes}}: Given a weighted set of
        vertical rays $\RaySet$ and a set of planes $\PlaneSet$ with
        capacities in $\Re^3$, find a maximum weight subset $\OSet$ of
        $\RaySet$ so that, for each plane $\plane$, the number of rays
        of $\OSet$ that intersect $\plane$ is at most
        $\capacityX{\plane}$.
        
        \item {\rm{\PackPointsInDisks}}: Given a set $\ObjSet$ of
        disks with capacities and a weighted set $\PntSet$ of points,
        find a maximum weight subset $\OSet$ of the points so that
        each disk $r \in \ObjSet$ contains at most $\capacityX{r}$
        points of $\OSet$.
    \end{inparaenum}
    
    \problab{prob:planes-points}
\end{problem}

Since the union complexity of halfspaces in three dimensions is
linear, we get the following from \thmref{p:hypergraph} (and the 3d
analogue of \lemref{clarkson}).

\begin{corollary}
    One can compute, in polynomial time, a constant factor
    approximation to the optimal solution of the \PackHalfspaces
    problem.
    
    \corlab{pack:halfspaces}
\end{corollary}

Standard point/plane duality implies that the same result holds for
the dual problem.  Namely, a point $(a,b,c)$ gets mapped to the plane
$z=ax+by-c$ and a plane $z=ax+by+c$ gets mapped to the point
$(a,b,-c)$.  Also, a point lies below a given plane if and only if the
dual point of the plane lies below the dual plane of the point. As
such, the dual of an instance of \PackHalfspaces is an instance of
\PackRaysInPlanes (and vice versa).

Thus, \corref{pack:halfspaces} implies the following.

\begin{corollary}
    One can compute, in polynomial time, a constant factor
    approximation to the optimal solution of the \PackRaysInPlanes
    problem.
\end{corollary}

Finally, observe that an instance of \PackPointsInDisks can be lifted
into an instance of \PackRaysInPlanes, by the standard lifting $f(x,y)
= (x,y,x^2+y^2)$, which maps points and disks in the plane to
halfspaces and points in three dimensions \cite{bcko-cgaa-08}.

\begin{corollary}
    One can compute, in polynomial time, a constant factor
    approximation to the optimal solution to the \PackPointsInDisks
    problem.
\end{corollary}

\subsection{Axis Parallel Rectangles/Boxes}

\subsubsection{Packing rectangles (2d)}

\begin{problem}{\rm{(\PackRectsInPnts.)}}
    Given a weighted set $\RectSet$ of axis-parallel rectangles in the
    plane, and a point set $\PntSet$ with capacities, find a maximum
    weight subset $\OSet \subseteq \RectSet$, such that, for any $\pnt
    \in \PntSet$, the number of rectangles of $\OSet$ containing
    $\pnt$ is at most $\capacityX{\pnt}$.
\end{problem}

Note that the union complexity of a set of rectangles can be
quadratic.  Hence we cannot simply use \thmref{p:regions:low:union} to
get a meaningful approximation.  However, by using the standard
approach for approximating the independent set of rectangles, one can
get a reasonable approximation as the following lemma testifies.

\begin{lemma}
    Given an instance $(\RectSet, \PntSet)$ of \PackRectsInPnts with
    $m$ rectangles, one can compute, in polynomial time, a subset
    $\OSet \subseteq \RectSet$ of total weight $\Omega( \Opt /\log m)$
    such that no capacity constraint of $\PntSet$ is violated, where
    $\Opt$ is the weight of the optimal solution.
    
    \lemlab{p:rect:in:pnts}
\end{lemma}

\begin{proof}
    It is straightforward to verify that a set of rectangles that
    intersect a common line have linear union complexity. Therefore it
    follows from \thmref{p:regions:low:union} that we can get a
    constant factor approximation for such instances. Given an
    arbitrary set of axis-parallel rectangles, we can reduce it to the
    case in which all rectangles intersect a common line as follows.
    
    We construct an interval tree on $\RectSet$. Let $\rect \in
    \RectSet$ be the median rectangle of $\RectSet$ when sorted by
    left edges.  Let $\Line$ denote the vertical line which passes
    through the left edge of $\rect$, and let $\RectSet_\Line$ denote
    the set of rectangles it intersects.  We associate
    $\RectSet_\Line$ with the root of our tree, and then recursively
    build left and right subtrees for the rectangles in
    $\RectSet\setminus \RectSet_\Line$ that lie to the left or right
    of $\Line$, respectively.  The recursion bottoms out once every
    rectangle has been stabbed by a line.  Clearly the depth of the
    tree is $O(\log m)$, since each time we choose the median line and
    only recursively continue on those rectangles that do not
    intersect it. Therefore there exists a level of the range tree
    that has a solution of weight $\Omega(\Opt / \log{m})$. The
    algorithm now considers each level of the tree separately, and for
    each node at the given level, it constructs an approximate
    solution for the rectangles associated with the node using the
    constant factor approximation algorithm guaranteed by
    \thmref{p:regions:low:union}. Next, the algorithm considers the
    union of all these solutions to form the solution for this level.
    Since two rectangles associated with two different nodes at the
    same depth in the range tree do not intersect, the resulting set
    is a valid solution.
    
    The algorithm returns the best solution found among all the
    levels. Clearly, its weight is $\Omega(\Opt / \log{m})$.
\end{proof}

\subsubsection{Packing axis-parallel boxes (3d)}

The union complexity of axis-parallel boxes in $\Re^3$ that contain a
common point is also linear, and therefore a similar approach as above
will enable us to solve the following problem.

\begin{problem}{\rm{(\PackBoxesInPnts.)}}
    Given a weighted set $\BoxSet$ of axis-parallel boxes in $\Re^3$,
    and a point set $\PntSet$ with capacities, find a maximum weight
    subset $\OSet \subseteq \BoxSet$, such that, for any $\pnt \in
    \PntSet$, the number of boxes of $\OSet$ containing $\pnt$ is at
    most $\capacityX{\pnt}$.
    
    \problab{pack:boxes}
\end{problem}

\begin{lemma}
    Given an instance $(\BoxSet, \PntSet)$ of \PackBoxesInPnts with
    $m$ boxes, one can compute, in polynomial time, a subset $\OSet
    \subseteq \RectSet$ of total weight $\Omega( \Opt /\log^3 m)$ such
    that no capacity constraint of $\PntSet$ is violated, where $\Opt$
    is the weight of the optimal solution.
    
    \lemlab{p:boxes:in:pnts}
\end{lemma}

\begin{proof}
    We build a multi-layer interval tree on the boxes. On the top
    layer, we build a balanced tree on the $x$-axis projection of the
    boxes, where a node $v_{x'}$ stores all boxes intersecting the
    plane $x=x'$. Next, we build for each such node a secondary
    interval tree on the $y$-axis projections, and for each node on
    this secondary data-structure we build a third layer
    data-structure on the $z$-axis projections. All the boxes are
    stored in the nodes of the third layer data-structure.
    
    First, observe that the boxes stored in a node on the third layer,
    $v_{x'y'z'}$, all contain the point $(x', y', z')$. Now, the union
    complexity of axis parallel boxes all sharing a common point is
    linear. As such, we can apply \thmref{p:hypergraph} to compute a
    packing that does not violate the capacities and is a constant
    factor approximation to the optimal (on this restricted set of
    boxes).  Now, for each third layer tree, we find the level that
    contains the best possible combined solution (taking the union of
    the solutions at a level is valid since there is no point in
    common to two boxes that are stored in two different nodes at the
    same level). This assigns values for each node in the secondary
    tree. Again, we choose for each secondary tree the level with the
    maximum total weight solution. We assign this value to the
    corresponding node in the first layer data-structure. Again, we
    choose the level with the highest possible value. This corresponds
    to a valid solution that complies with the capacity constraints.
    
    As for the quality of approximation, observe that for a tree at a
    given layer, at least a logarithmic factor of the remaining weight
    of the optimal solution is contained in some level of the tree.
    As such, each time we go down one layer in this data-structure we
    lose at most a logarithmic factor of the optimal solution, and
    hence the quality of approximation of this algorithm is
    $\Omega(\log^3 m)$.
\end{proof}

\begin{remark}
    \begin{inparaenum}[(A)]
        \item It is natural to ask if this result can be extended to
        higher dimensions. However, it is easy to see that the union
        complexity of $n$ axis-parallel boxes in four dimensions that
        all contain a common point can be quadratic. As such, this
        approach would fail miserably. We leave the problem of getting
        a better approximation for this case as an open problem for
        further research.
        
        \item The continuous version seems to be considerably
        easier. For the weighted case (with unit capacities; that is,
        the independent set variant) an $O(\log^{d-1} m / \log \log
        m)$ approximation is known \cite{ch-aamis-09} by solving the
        two dimensional case, and then using the above interval tree
        technique to apply it for higher dimensions. For the
        unweighted continuous case, an $O(\log \log m \log^{d-2} m )$
        approximation is known \cite{cc-misr-09}.
        
        The discrete version is different than the continuous version
        because, for example, if considering boxes in $\Re^3$ that all
        intersect the $xy$-plane, the induced two dimensional instance
        fails to encode the capacity constraints, as they rises from
        points in three dimensions that do not lie on this plane
        (while in the continuous case, it is enough to solve the
        induced problem in this plane).
    \end{inparaenum}
\end{remark}

\subsubsection{Packing points into rectangles}

\begin{problem}{\rm{(\PackPntsInRects.)}}
    Given a weighted set $\PntSet$ of points and a set $\RectSet$ of
    axis-parallel rectangles with capacities in the plane, find a
    maximum weight subset $\OSet \subseteq \PntSet$, such that, for
    any $\rect \in \RectSet$, the number of points of $\OSet$
    contained in $\rect$ is at most $\capacityX{\rect}$.
    \problab{pnts:in:rects}
\end{problem}

We first observe that the hypergraph that arises from a instance
$\hgraph = (\PntSet,\RectSet)$ of \PackPntsInRects, might not have the
bounded growth property for any reasonable growth function.  To see
this consider the following example.

\parpic[r]{%
   \begin{minipage}{0.4\linewidth}%
       \includegraphics[width=.9\linewidth]{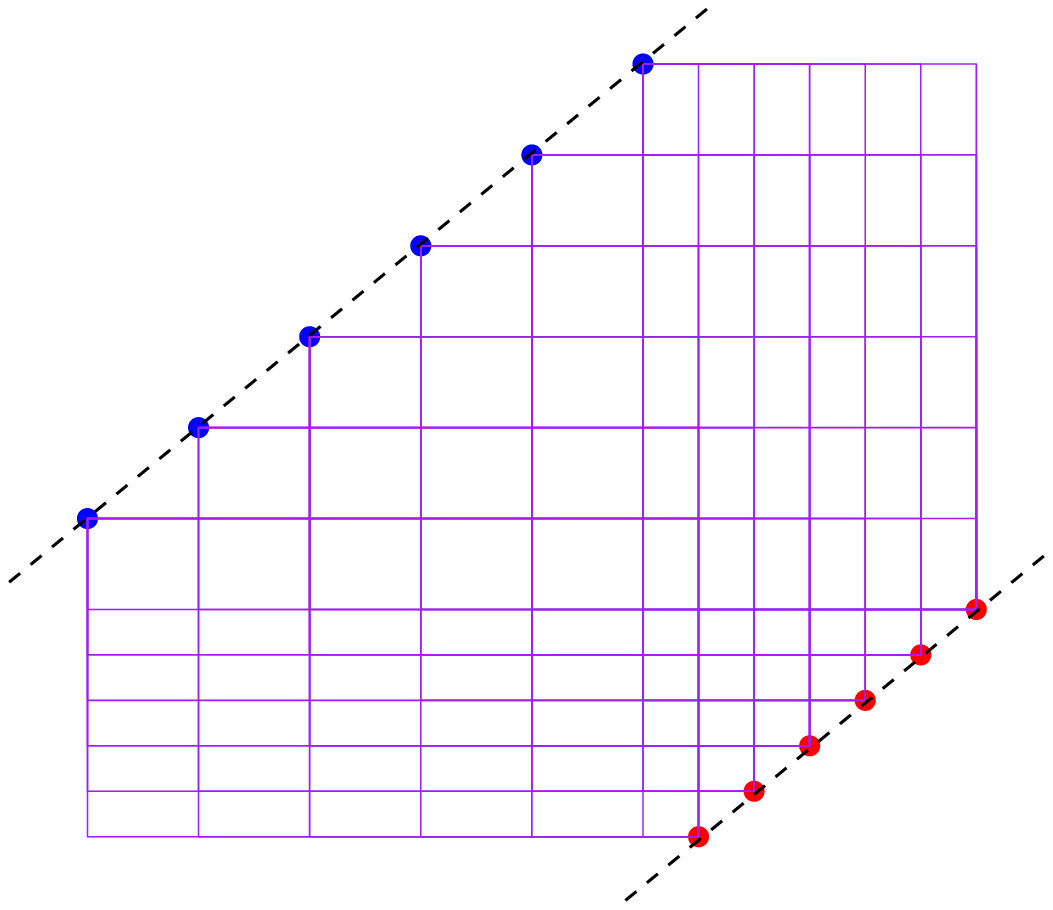}%
   \end{minipage}}

Consider two parallel lines in the plane with positive slope.  Place
$n/2$ points on each line such that all the points on the top line lie
above and to the left of all the points on the bottom line.  Let the
set of rectangles for this instance of \PackPntsInRects be all the
rectangles which have a point on the top line as their upper left
corner and a point on the bottom line as their lower right corner.  In
this case any subset of $O(t)$ points from the top line and $O(t)$
points from the bottom line induce a set of $O(t^2)$ hyperedges, each
of size 2.  Therefore, $\feq{1}{t} = \Omega\pth{t^2}$, and hence
\thmref{p:hypergraph} only gives an $O(\Energy)$ approximation.

Since we cannot hope to apply our main result to the case of
\PackPntsInRects, we will instead seek a bi-criterion approximation.
Our algorithm here is inspired by the work of Ezra \etal
\cite{aes-ssena-10} on $\eps$-nets for rectangles.  Before tackling
this problem, we will first consider an easier variant, which will be
useful later in obtaining a bi-criterion approximation.  In the
following, we call a set of rectangles such that all their (say)
bottom edges lies on a common line a \emphi{skyline}.

\begin{problem}{\rm{(\PackPntsInSkyline.)}}
    Given a weighted set $\PntSet$ of points and a set $\RectSet$ of
    skyline rectangles with capacities in the plane, find a maximum
    weight subset $\OSet \subseteq \PntSet$, such that, for any $\rect
    \in \RectSet$, the number of points of $\OSet$ contained in
    $\rect$ is at most $\capacityX{\rect}$.
\end{problem}

\begin{lemma}
    Let $\PntSet$ be a set of $n$ points in the plane all placed above
    the $x$-axis. Let $\feq{k}{n}$ be the maximum number of different
    subsets of $\PntSet$ of size $k$ that are realized by intersecting
    $\PntSet$ with a rectangle whose bottom edge lies on the
    $x$-axis. We have that $\feq{k}{n} = O(n k^2)$.
    
    \lemlab{few:light:rectangles}
\end{lemma}
\begin{proof}
    Consider a rectangle $\rect$ with its bottom edge lying on the
    $x$-axis, and which contains $k$ points of $\PntSet$. Lower its
    top edge till it passes through a point of $\PntSet$, and let
    $\pnt$ denote this point. Similarly, move its left and right edges
    till they pass through points of $\PntSet$. Let $\rect'$ be this
    new canonical rectangle.  Now, let $\ileft$ (resp. $\iright$) be
    the number of points of $\PntSet$ inside $\rect'$ that are to the
    left (resp. right) of $\pnt$. Clearly, $\pth[]{\pnt, \ileft,
       \iright}$ uniquely identifies this canonical rectangle. This
    implies the claim as $\pnt \in \PntSet$, $\ileft \leq k$ and
    $\iright \leq k$, and hence the numbers of such triples is $O\pth{
       \cardin{\PntSet} k^2}$.
\end{proof}

\begin{lemma}
    Given an instance of \PackPntsInSkyline, one can compute, in
    polynomial time, an $O( 1 )$-approximation to the optimal
    solution.
\end{lemma}

\begin{proof}
    Consider the associated hypergraph $\hgraph = (\VSet,\HESet)$.  By
    \lemref{few:light:rectangles}, this hypergraph has the bounded
    growth property with $\feq{k}{t} = t O(k^2)$ (here $\gamma(t) =
    1$). Therefore, the algorithm of \thmref{p:hypergraph} provides
    the required approximation.
\end{proof}

\begin{lemma}
    Given a set $\PntSet$ of $n$ points in the plane, and a parameter
    $k$, one can compute a set $\RectSetA$ of $O\pth{ k^2 n\log n }$
    axis-parallel rectangles, such that for any axis-parallel
    rectangle $\rect$, if $\cardin{\rect \cap \PntSet} \leq k$, then
    there exists two rectangles $\rect_1, \rect_2 \in \RectSetA$ such
    that $\pth[]{\rect_1 \cup \rect_2} \cap \PntSet = \rect \cap
    \PntSet$.
    
    Furthermore, consider the graph where two points of $\PntSet$ are
    connected if they belong to the same rectangle in
    $\RectSetA$. Then the number of edges in this graph is $O( nk \log
    n)$.
    
    \lemlab{rect:set}
\end{lemma}
\begin{proof}
    Find a horizontal line $\Line$ that splits $\PntSet$ equally, and
    compute all the skyline rectangles that contain at most $k$ points
    of $\PntSet$ (that is, compute both the rectangle above and below
    the line).  By \lemref{few:light:rectangles}, the number of such
    rectangles is $O( n k^2)$. Now, recursively compute the rectangle
    set for the points above $\Line$, and for the points below
    $\Line$. Clearly, the number of rectangles generated is $O(k^2 n
    \log n)$, and let $\RectSetA$ denote the resulting set of
    rectangles.
    
    Now, consider any axis-parallel rectangle $\rect$ such that
    $\cardin{\rect \cap \PntSet}\leq k$. If it is does not intersect
    $\Line$ then by induction it has the desired property. Otherwise,
    if $\rect$ intersects $\Line$, then it can be decomposed into two
    skyline rectangles, each one of them contains at most $k$ points
    of $\PntSet$. By construction, for each of these rectangles there
    is a rectangle in $\RectSetA$ that contains exactly the same set
    of points.
    
    As for the second claim, we apply a similar argument.  Consider an
    edge $\pnt \pntA$ in this graph that arise because of a top
    skyline rectangle of $\Line$. Furthermore, assume that $\pnt$ is
    higher than $\pntA$ and to its right. Clearly, there are at most
    $k$ such edges emanating from $\pnt$, as the skyline rectangle
    having $\pnt$ as its top right corner and having its left edge
    through $\pntA$ contains at most $k$ points, and each such
    rectangle corresponds to a unique edge. As such, we get that the
    number of edges in the graph is $E(n) = O(nk) + 2T(n/2) = O(n k
    \log n)$.
\end{proof}

\begin{remark}
    A slightly more careful analysis shows that the number of
    rectangles in the set computed by \lemref{rect:set} that contain
    exactly $k$ points is $O(n k \log n)$. This will not be needed for
    our analysis.
\end{remark}

\begin{remark}
    Consider an instance $\hgraph = (\PntSet,\RectSet)$ of
    \PackPntsInRects.  Let $\hgraph' = (\PntSet,\RectSetA)$ be a
    modified instance of \PackPntsInRects where $\RectSetA$ is
    obtained from $\RectSet$ by replacing each rectangle by two new
    rectangles whose union covers that same set of points.
    \lemref{rect:set} guarantees that this can be done such that
    $\cardin{\RectSetA} = O\pth{n^3\log n}$.  One might be tempted to
    believe that we can plug $\hgraph'$ into \thmref{p:hypergraph} in
    order to get a bi-criteria approximation for $\hgraph$.
    Unfortunately, this does not work since (as the following example
    shows) the hypergraph does not have the bounded growth property
    for any meaningful growth function.
    
    \parpic[r]{%
       \begin{minipage}{0.45\linewidth}%
           \includegraphics[width=.95\linewidth]{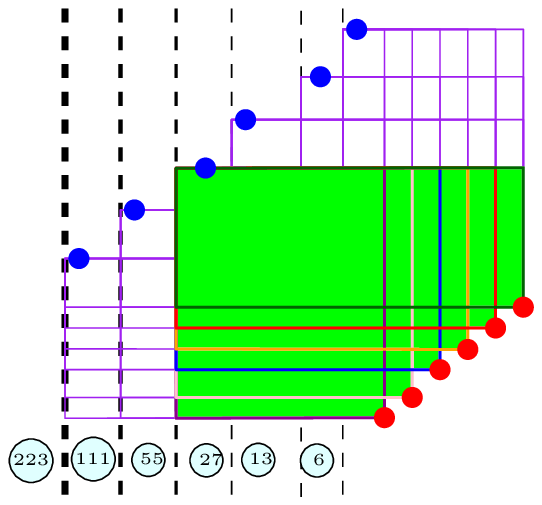}%
       \end{minipage}}
    
    Consider two parallel lines in the plane with positive slope.
    Place $\Theta(\log n)$ points of $\PntSet$ on each line (or close
    to the line) such that the points on (or close to) the top line
    all lie above and to the left of those on the bottom line.  The
    remaining points of $\PntSet$ will all lie below the points on the
    diagonals.  Let $T$ be the interval tree of $\PntSet$ (using
    vertical split lines).  Specifically, the remaining points of
    $\PntSet$ will be placed such that each point on the top line lies
    in a different level of $T$, and all the points on the bottom line
    lie in the same node as the rightmost point on the top line.  More
    specifically, the leftmost point on the top line will correspond
    to the root and the points in order from the left to right on the
    top line will correspond to continually walking down in the tree.
    (A $\textcircled{u}$ in the figure represents a cluster of $u$
    points close together.)  Now let $X$ be the subset of $\PntSet$
    which consists of the two set of $\Theta(\log n)$ points on the
    diagonal lines.  Consider the intersection sub-hypergraph induced
    by $X$.  Suppose that $\RectSet$ has rectangle for every pair of
    points in $\PntSet$ that can be obtained as the intersection of a
    rectangle with $\PntSet$.  Then any pair of points from the top
    and bottom diagonals will correspond to a hyperedge in this
    induced sub-hypergraph.  Therefore, $\feq{1}{\log n} =
    \Omega\pth{\log^2 n}$, and hence \thmref{p:hypergraph} only gives
    an $O(\Energy)$ approximation.
\end{remark}

Since (as the above remark demonstrates) we cannot directly apply
\lemref{rect:set}, our approach will be more roundabout.  We first
show how to solve the independent set variant of our problem (i.e.,
unit capacities). Next, we slice the rectangles of the given instance
with non-uniform capacities case into subrectangles with unit
capacities, and plug it into the above algorithm to get a meaningful
approximation.

\begin{lemma}
    Given an instance of $\hgraph = (\VSet,\RectSet)$ of
    \PackPntsInRects with unit capacities, one can compute a subset
    $\VSetA \subseteq \VSet$, such that the total weight of $\VSetA$
    is $\Omega(\Opt/ \log \Energy)$ and each rectangle of $\RectSet$
    contains at most $2$ points of $\PntSet$, where $n =
    \cardin{\PntSet}$.
    
    \lemlab{first}
\end{lemma}

\begin{proof}
    We first use \lemref{sparsify} to sparsify the given instance. We
    now have a set of $\PntSet \subseteq \VSet$ of $t = \Theta(
    \Energy \log \Energy)$ points, and an associated fractional
    solution, such that none of the constraints are violated. The
    value of the fractional solution on $\InducedX{\hgraph}{\PntSet}$
    is $\Omega(\Opt)$, and as such we restrict our search for a
    solution to $\PntSet$.
    
    Furthermore, we can assume that the value assigned to each point
    of $\PntSet$ by this fractional solution is exactly $1/M$ (we
    replicate a point $i$ times if it is assigned value $i/M$), where
    $M = O( \log \Energy)$. Note, that none of the rectangles of
    $\RectSet$ contains more than $M$ points of $\PntSet$. In
    particular, by \lemref{rect:set}, one can build a set of
    rectangles $\RectSetA$ of size $O(M^2 t \log t)$, such that every
    rectangle of $\RectSet$ can be covered by the union of two
    rectangles of $\RectSetA$; formally, for every $\rect \in
    \RectSet$ there exists $\rect_1, \rect_2 \in \RectSetA$ such that
    $\rect \cap \PntSet = \pth[]{\rect_1 \cup \rect_2} \cap
    \PntSet$. We build a conflict graph $\graph$ over $\PntSet$
    connecting two points if (i) they are both contained in a
    rectangle of $\RectSet$, and (ii) there is a rectangle of
    $\RectSetA$ that contains them both. By \lemref{rect:set} this
    graph has at most $O(M t \log t) = O( \Energy \log^3 \Energy )$
    edges and $t = \Theta( \Energy \log \Energy)$ vertices.
    
    We further add edges to $\graph$ making a clique out of each group
    of duplicated points that arose from a single given point of
    $\PntSet$ (this is needed since when duplicating the points we
    perturbed them in order to maintain the implicit general position
    assumptions of \lemref{rect:set}, and one needs to guarantee that
    at most one of these copies is picked to the independent set).
    Now for a point $\pnt \in \PntSet$ with \LP value $x_\pnt$, the
    number of duplicated points is $x_\pnt M$.  Hence the number of
    edges added for these cliques is
    \[
    \sum_{\pnt \in \PntSet} \binom{x_\pnt M}{2}%
    \leq%
    \sum_{\pnt \in \PntSet} x_\pnt^2M^2%
    \leq%
    M^2\sum_{\pnt \in \PntSet}x_\pnt%
    \leq%
    M^2\Energy = O\pth{ \Energy \log^2 \Energy },
    \]
    and hence the number of edges in $\graph$ overall is $O(\Energy
    \log^3 \Energy)$.
    
    It is easy to verify that $\graph$ has average degree $O( \log^2
    \Energy)$, and the total weight of the vertices is $\Theta( \Opt
    \log \Energy)$, as such, by \Turan's theorem, one can compute an
    independent set of vertices in this graph of weight $\Omega(
    w(\PntSet)/ (\text{average degree} + 1))$ $ = \Omega( \Opt/\log
    \Energy)$.

    Now, it is easy to verify that any rectangle in $\RectSet$
    contains at most two points of this independent set.
\end{proof}

\begin{theorem}
    Given an instance of $(\VSet,\RectSet)$ of \PackPntsInRects (with
    arbitrary capacities), one can compute in polynomial time a subset
    $\VSetA \subseteq \VSet$ that is an $\pth[]{O( \log \Energy),
       2}$-approximation to the optimal solution.
    
    
    \thmlab{pack:points:in:rects}
\end{theorem}
\begin{proof}
    Compute a fractional solution to the given instance.  Split each
    rectangle $\rect$ with capacity $\capacityX{\rect}$ into
    $\ceil{\capacityX{\rect}/3}$ rectangles, each one containing at
    most value $4$ from the fractional solution (this can be done by
    sweeping the rectangle from left to right, and splitting it
    whenever the fractional solution inside the current portion
    exceeds $3$). Consider now a unit capacity instance on the same
    point set but with these new rectangles.  We use \lemref{first} in
    order to get an $\pth[]{O( \log \Energy), 2}$-approximation for
    this new instance.
    
    We now show that this solution we obtained for the unit capacity
    instance is also an $\pth[]{O( \log \Energy), 2}$-approximation to
    the original instance.  First observe that the \LP value on this
    new instance is $\Omega\pth{\Opt}$ (where $\Opt$ is the \LP value
    of the original instance) since scaling down the fractional
    solution to original instance by a factor of 4 would be a valid
    solution to the \LP for the new instance (since these newly
    created rectangles each contained at most 4 from the fractional
    solution), and hence the weight of the approximation is
    $\Omega\pth{\Opt/\log\Energy}$.  Furthermore, we know every
    rectangle $\rect \in \RectSet$ contains at most $2
    \ceil{\capacityX{\rect}/3} \leq \max\pth{ 2, \capacityX{\rect} }$
    points from this solution, since in the new instance each
    rectangle from $\RectSet$ was replaced with
    $\ceil{\capacityX{\rect}/3}$ rectangles each of which contains at
    most two points from the computed solution. (Note that the
    inequality holds since $\capacityX{\rect}$ has integral value.)
\end{proof}


\section{Packing points into fat triangles}
\seclab{p:points:in:triangles}

In this section, we give a bi-criterion approximation for packing
points into a set of $\alpha$-fat triangles. More precisely, we
consider the following problem.

\begin{problem}{\rm{(\PackPntsInFTri.)}}
    Given a weighted set $\PntSet$ of points and a set $\TriSet$ of
    $\alpha$-fat triangles in the plane such that each triangle $\Tri$
    has a capacity $\capacityX{\Tri}$, find a maximum weight subset
    $\OSet \subseteq \PntSet$, such that, for each $\Tri \in \TriSet$,
    the number of points of $\OSet$ contained in $\Tri$ is at most
    $\capacityX{\Tri}$.  \problab{pnts:in:f:tri}
\end{problem}

The approximation algorithm uses the following building blocks:
\begin{compactenum}[$~$~~~(A)]
    \item We prove that, for a given point set, there exists a small
    number of canonical sets such that for any fat triangle that
    covers at most $k$ points, there exists a constant number of these
    canonical sets whose union covers exactly the same points.
    Showing this result is quite technical and requires non-trivial
    modifications of the work of Aronov \etal \cite{aes-ssena-10} (in
    particular, their work does not imply this result). This is
    delegated to \secref{canonical:sets}, see \thmref{canonical:fat}
    for the exact result.
    
    \item An algorithm for approximating the unit capacity case. This
    follows by an algorithm similar to the one in \lemref{first}, see
    \lemref{f:t:first} for details. Note that this uses the result
    from (A) to get the required approximation.
    
    \item A partition scheme that shows that a fat triangle (with a
    measure defined over it) can be ``partitioned'' into $O(k)$
    triangles such that any triangle in this partition has measure at
    most $1/k$; see \lemref{f:triangle:m:p}.
\end{compactenum}

\smallskip Putting these components together yields the approximation
algorithm; see \thmref{pack:points:in:triangles} for details.

\subsection{The unit capacity case}

\begin{lemma}
    Given an instance $\hgraph = (\VSet,\TriSet)$ of \PackPntsInFTri
    with unit capacities, one can compute a subset $\VSetA \subseteq
    \VSet$ such that the total weight of $\VSetA$ is $\Omega\pth{
       \Opt/ \log^6 \Energy }$ and each triangle of $\TriSet$ contains
    at most $\FTriConst$ points of $\PntSet$.
    
    \lemlab{f:t:first}
\end{lemma}

\begin{proof}
    We follow the proof of \lemref{first}.  We first use
    \lemref{sparsify} to sparsify the given instance. We now have a
    set of $\PntSet \subseteq \VSet$ of $t = \Theta( \Energy \log
    \Energy)$ points and a corresponding fractional solution that is
    feasible. The value of the fractional solution on
    $\hgraph_\PntSet$ is $\Omega(\Opt)$, and as such we restrict our
    search for a solution to $\PntSet$.
    
    Furthermore, we can assume that the value assigned to each point
    of $\PntSet$ by this fractional solution is exactly $1/M$ --- we
    replicate a point $i$ times if it is assigned value $i/M$ ---
    where $M = O( \log \Energy)$. Note that none of the triangles of
    $\TriSet$ contains more than $M$ points of $\PntSet$. In
    particular, by \thmref{canonical:fat}, one can construct a set
    $\FamilyA$ of regions of size $O\pth{ M^3 t\log^2 t }$ such that,
    for every triangle of $\Tri \in \TriSet$, there exists a subset
    $\{\hedgeA_1, \ldots, \hedgeA_k\} \subseteq \FamilyA$ of at most 9
    regions (i.e. $k\leq 9$) such that $\PntSet \cap \Tri = \PntSet
    \cap \pth{\cup_{i=1}^k \hedgeA_i}$.  We build a conflict graph
    $\graph$ over $\PntSet$ connecting two points if (i) they are both
    contained in a triangle of $\TriSet$, and (ii) there is a set of
    $\FamilyA$ that contains both of them. Since the number of sets in
    $\FamilyA$ is $O\pth{ M^3 t\log^2 t }$, and each such set has size
    at most $M$, it follows that the number of edges in the resulting
    graph $\graph$ is $O\pth{ M^5 t\log^2 t }$, and the number of
    vertices is $t = \Theta( \Energy \log \Energy)$ vertices.  As in
    the proof of \lemref{first}, we also add edges between replicated
    points (since these edges do not affect our analysis, we ignore
    them for the sake of simplicity of exposition).
    
    The graph $\graph$ has average degree $O\pth{ M^5 \log^2 t } =
    O\pth{ \log^7 \Energy}$, and the total weight of the vertices is
    $\Theta( \Opt \log \Energy)$. Therefore, by \Turan's theorem, one
    can compute an independent set of vertices in this graph of weight
    $\Omega( w(\PntSet)/ (\text{average degree} + 1))$ $ = \Omega(
    \Opt/\log^6 \Energy)$.
    
    Finally, it is easy to verify that any triangle in $\TriSet$
    contains at most $\FTriConst$ points of this independent set.
\end{proof}

\subsection{Covering a measure on a fat triangle}

At this point we would like to use \lemref{f:t:first} in order to get
a bi-criteria approximation for the case in which the capacities are
arbitrary, as we did in \thmref{pack:points:in:rects}.  However, doing
so directly proves more challenging for fat triangles than axis
parallel rectangles.  This is because our general procedure requires
that, given an object $x$ of a given type such that the total
fractional value of the points in $x$ is non-zero, we need to be able
to decompose $x$ into $O(\capacityX{x})$ smaller objects of the same
type such that, for each smaller object, the fractional value of the
points in the object is only a constant.  This can easily be done for
axis parallel rectangles by using vertical splitting lines (as was
done in \thmref{pack:points:in:rects}), but it is more challenging for
fat triangles. However, the following lemma shows that such a
decomposition is still possible for fat triangles.

\newcommand{\constTriCoverM}{18}

\begin{lemma}
    Let $\mu$ be a measure defined over the plane, and consider a fat
    triangle $\Tri$. Then, for any integer $k$, one can cover $\Tri$
    by at most $\constTriCoverM k$ fat triangles, such that the
    measure of each of these triangles is at most $\mu(\Tri)/k$.
    
    \lemlab{f:triangle:m:p}
\end{lemma}
\begin{proof}
    To simplify the presentation we assume that $\mu(\Tri) = 1$.  We
    recursively build a tree on $\Tri$ by partitioning the original
    triangle $\Tri$ into $4$ similar triangles as shown in
    \figref{t:partition}. Each node of this tree corresponds to a
    triangle from this recursive construction. We stop the recursive
    partition for a node $v$ as soon as the measure of the triangle
    $\Tri_v$ associated with it is at most $1/k$.
    
    \parpic[r]{%
       \begin{minipage}{0.2\linewidth}%
           \includegraphics{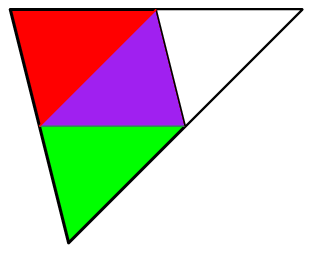}%
           \vspace{-1.0cm}%
           \caption{\hspace*{-1.3cm}}%
           \figlab{t:partition}
       \end{minipage}}
    
    Once we have the tree, we select a set $S$ of nodes of the tree as
    follows. We find the lowest node $v$ in the tree such that the
    measure of its corresponding triangle is at least $1/k$. We add
    the node $v$ to $S$ and we treat the measure inside the triangle
    corresponding to $v$ as being $0$. We repeat this process until
    the measure left uncovered is smaller than $1/k$, at which point
    we take the lowest node covering the remaining measure and we add
    it to $S$. We also add the root of the tree to $S$. Note that the
    set $S$ contains at most $k + 1$ nodes: since the initial measure
    is one, we added at most $k$ nodes to $S$ that are not the root.
    \parpic[r]{%
       \begin{minipage}{0.25\linewidth}%
           \includegraphics{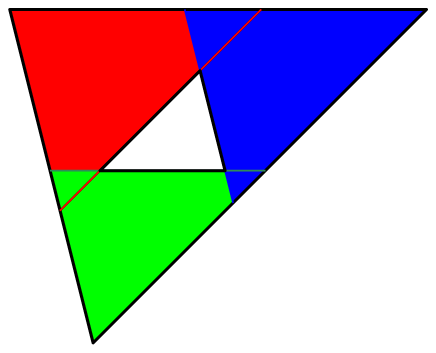}%
           \vspace{-1.cm} \caption{\hspace{-1.7cm}}%
           \figlab{triangle:hole}
       \end{minipage}}
    
    We also add all the nodes in this tree that are the least common
    ancestor (\LCA) of a pair of nodes in $S$. Let $S'$ be the
    resulting set of nodes. Now we can show that, for any tree $T$ and
    any subset $R$ of nodes of $T$, the set of all \LCA nodes of all
    of the pairs of nodes in $R$ has size at most $|R| -
    1$\footnote{Let $u, v$ be the pair of nodes in $R$ whose \LCA has
       maximum depth. Let $z$ be the \LCA of $u$ and $v$, and let $R'
       = R - \{u, v\} \cup \{z\}$. The pairs in $R$ and the pairs in
       $R'$ have the same set of \LCA nodes and thus the number of
       \LCA nodes of a set of size $r$ satisfies the recurrence $f(r)
       \leq f(r - 1) + 1$ and $f(2) = 1$.}.  Therefore there are at
    most $|S| - 1$ nodes in $S' - S$ and thus the size of $S'$ is at
    most $2|S| \leq 2(k + 1)$.
    Let $\TriSet$ be the set of triangles induced by the triangles
    corresponding to the children of the nodes of $S'$ (i.e., every
    node of $S'$ gives rise to four triangles).  Consider the
    partition of the original triangle formed by $\TriSet$. It is easy
    to verify that every face in this arrangement has measure at most
    $1/k$, and the number of faces of this arrangement, denoted by
    $n$, is at most $8k$ (observe that since the triangles arise out
    of a recursive partition, a pair of such triangles is either
    disjoint or contained in each other). Furthermore, each face of
    this arrangement is either a triangle (which is a scaled and
    rotated copy of the original triangle), or the difference of two
    triangles where one contains the other (having this property is
    why we took the children of every \LCA). We will refer to a face
    which is the difference of two triangles as an annulus face.  Let
    $n'$ be the number of annulus faces in this arrangement.  Clearly,
    $n' \leq \cardin{S'}=2k$, as one can charge an annulus face to the
    node of $S'$ that induced the hole in this face.
    
    \parpic[r]{%
       \begin{minipage}{0.25\linewidth}%
           \includegraphics{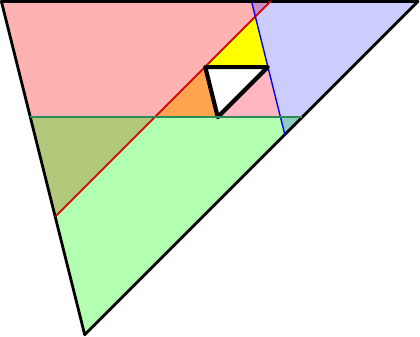}%
           \vspace{-1.01cm}%
           \caption{\hspace{-1.7cm}}%
           \figlab{triangle:hole:2}
       \end{minipage}}
    
    We claim that an annulus face can be covered by the union of six
    translated and rotated copies of the original triangle.  The easy
    case, is when the hole is a scaled and translated reflection of
    the outside triangle, see \figref{triangle:hole}, where three
    triangles are sufficient. The other case is when the hole is a
    scaled translated copy of the outer triangle, which by attaching
    to the hole three translated and rotated copies of the hole, gets
    reduced to the other case, see \figref{triangle:hole:2}. (The
    property used here implicitly is that the outer triangle of the
    annulus can be partitioned into translated and rotated copies of
    the hole triangle, as the hole arises out of a recursive partition
    of the outer triangle.)
    
    As such, there are $n-n'$ triangular faces in this arrangement and
    $n'$ annulus faces. Thus, the original triangle can be covered by
    $(n - n') + 6n' = n + 5n' \leq 8k + 10 k \leq \constTriCoverM k$
    translated and scaled copies of the original triangle that cover
    it completely, and no triangle in this collection has measure that
    exceeds $1/k$.
\end{proof}

\begin{remark}
    Given a weighted set of $n$ points defining the measure inside the
    given fat triangle, the cover of \lemref{f:triangle:m:p} can be
    computed in $O(n \log n)$ time. This requires using known
    techniques used in constructing compressed quadtrees, see
    \cite{h-gaa-11} for details.
    
    \remlab{algorithmic:f:t:p}
\end{remark}

\subsection{The result}

\begin{theorem}
    Given an instance of $(\VSet,\TriSet)$ of \PackPntsInFTri (with
    arbitrary capacities), one can compute, in polynomial time, a
    subset $\VSetA \subseteq \VSet$ that is $\pth[]{O\pth{ \log^6
          \Energy}, \FTriConst }$-approximation to the optimal
    solution.
    
    
    \thmlab{pack:points:in:triangles}
\end{theorem}

\begin{proof}
    Compute a fractional solution to the given instance.  For any
    triangle $\Tri$ in the plane, we denote by $\EnergyX{\Tri} =
    \sum_{\pnt \in \VSet \cap \Tri} x_\pnt$ the total mass of the
    fractional solution inside $\Tri$.  Next, we get a constant
    capacity instance out of $(\VSet,\TriSet)$ by replacing each
    triangle of $\TriSet$ by a ``few'' triangles covering it, such
    that the total mass of the fractional solution inside each of
    these new triangles is at most $c = 4 \cdot \constTriCoverM \cdot
    \FTriConst$. Formally, consider a triangle $\Tri \in \TriSet$, and
    let $k = \ceil{\capacityX{\Tri}/c}$.  If $\capacityX{\Tri} \leq c$
    then there is nothing to do (as $\EnergyX{\Tri} \leq
    \capacityX{\Tri} \leq c$), so we assume that $\capacityX{\Tri} >
    c$.  Applying the algorithmic version of \lemref{f:triangle:m:p},
    see \remref{algorithmic:f:t:p}, we cover $\Tri$ with at most
    \[
    \constTriCoverM k%
    =%
    \constTriCoverM \ceil{\frac{\capacityX{\Tri}}{c}}%
    =%
    \constTriCoverM \ceil{\frac{\capacityX{\Tri}}{ 4 \cdot
          \constTriCoverM \cdot \FTriConst}}%
    \leq%
    \ceil{\frac{\capacityX{\Tri}}{2\cdot \FTriConst}}
    \]
    triangles, where the total mass of the fractional solution inside
    each of them is at most $\EnergyX{\Tri}/k \leq \capacityX{\Tri}/
    \ceil{\capacityX{\Tri}/c} \leq c$.
    
    
    Now, consider the generated instance with these new triangles,
    where each such triangle has capacity one. To this end, scale down
    the solution of the \LP by a factor of $c$. Clearly, we now have a
    uniform capacity instance with an associated (valid) fractional
    solution having value $\Omega(\Opt)$ (where $\Opt$ is the optimal
    \LP value for the original instance). Furthermore, any solution to
    this unit capacity instance, would correspond to a solution to the
    original instance (since we covered each original triangle with at
    most $18k\leq \ceil{\frac{\capacityX{\Tri}}{2\cdot
          \FTriConst}}\leq \capacityX{\Tri}$ new unit capacity
    triangles).  Plugging this instance into \lemref{f:t:first} yields
    the required approximation. Specifically, every triangle $\Tri \in
    \TriSet$ contains at most $ \FTriConst \ceil{\capacityX{\Tri}/ ( 2
       \cdot \FTriConst )} \leq \max\pth{ \FTriConst, \capacityX{\Tri}
    }$ points of the computed set of points.
\end{proof}

\subsection{Canonical decomposition for fat triangles}
\seclab{canonical:sets}

In this section, we show that given a set $\PntSet$ of $n$ points in
the plane, and a parameter $k$, one can compute a set $\TriSetA$ of
$O\pth{ k^3 n\log^2 n }$ regions, such that for any $\alpha$-fat
triangle $\Tri$, if $\cardin{\Tri \cap \PntSet} \leq k$, then there
exists (at most) $\FTriConst$ regions in $\TriSetA$ whose union has
the same intersection with $\PntSet$ as $\Tri$ does.

Our construction follows closely the argumentation of Aronov \etal
\cite{aes-ssena-10}. However, our construction is (somewhat) different
and (arguably) simpler since we are considering a ``dual'' problem to
theirs. In particular, since modifying Aronov \etal
\cite{aes-ssena-10} to get our result is not obvious, we present it
here in detail.

\subsubsection{Initial setup}

To construct the set of regions, $\TriSetA$, we will use an approach
similar to that of \lemref{rect:set}.  As was observed in
\cite{aes-ssena-10}, we can restrict our attention to axis aligned
right triangles whose hypotenuse differs by no more than say one
degree from $-45^\circ$, as measured from the positive $x$-axis
(i.e. it is near isosceles and faces to the right).  In the following,
let $\Tri$ be an arbitrary such triangle that contains at most $k$
points.

\parpic[r]{%
   \begin{minipage}{0.3\linewidth}%
       \includegraphics[width=0.98\linewidth]{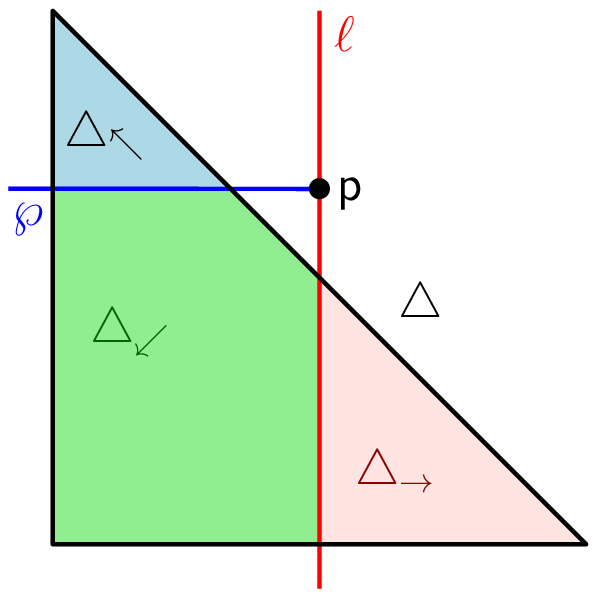}%
       \caption{Decomposing $\Tri$.}  \figlab{triangle:cut}
   \end{minipage}}

We first construct a two level interval tree on $\PntSet$, where the
first level partitions the points based on their $x$-coordinate, and
the second level based on their $y$-coordinate (and the splitting line
for each node goes through the median point).  Let $v$ be the highest
node in the first level of the interval tree whose corresponding split
line, $\Line$, intersects $\Tri$.  Let $\TriL$ and $\TriR$ denote the
portion of $\Tri$ to the left or right of $\Line$, respectively. Also,
let $u$ be the highest node in the second level tree rooted at the
left child of $v$ whose corresponding split line, $\LineA$, intersects
$\TriL$, and let $\TriLT$ and $\TriLB$ denote the portion of $\TriL$
above or below $\LineA$, respectively.  (Note that we may assume that
there exists split lines $\Line$ and $\LineA$ that intersect $\Tri$
and $\TriL$, respectively, since such regions that contain no points
can be skipped).  In the following, let $\pnt$ denote the point of
intersection between $\Line$ and $\LineA$. See \figref{triangle:cut}.

We now construct sets of canonical regions, $\TriSetR$, $\TriSetLT$,
and $\TriSetLB$ such that for any choice of $\Tri$ there exists
constant number of regions $\regionA_1, \ldots, \regionA_m$ in
$\TriSetR \cup \TriSetLT \cup \TriSetLB$, such that $\Tri \cap \PntSet
= \bigcup_i \pth{ \regionA_i\cap\PntSet }$, and $m\leq \FTriConst$.

We achieve this by showing that in each case (i.e., $\TriLT, \TriLB$
and $\TriR$) the region $\regionA$ under consideration can be
transformed into a polygonal region with a constant number of points
of $\PntSet$ (or orientations) defining its bounding edges, and whose
intersection with $\PntSet$ is the same as $\regionA$\footnote{For
   each point on the bounding edges, we will need to specify whether
   it is inside or outside the canonical region.  This can be encoded
   by a string of length $c$, where $c$ is some constant bounding the
   number of defining boundary points.  Hence we can specify the
   inclusion or exclusion of the boundary points while only increasing
   the number of canonical regions by a factor of $2^c = O(1)$, and
   hence we will not need to worry about such issues.}.  We then show
that the number of such regions needed for a particular choice $\Line$
and $\LineA$ is $O(nk^3)$.

In the following, let $\PntSet_v$ be the subset of points of $\PntSet$
stored in the subtree rooted $v$, and let $\PntSet_{u,v}$ be the set
of points stored in the subtree rooted at $u$.

\subsubsection{Handling the right portion of the triangle ($\TriR$)}


\parpic[r]{%
   \begin{minipage}{0.3\linewidth}%
       \includegraphics[width=0.98\linewidth]{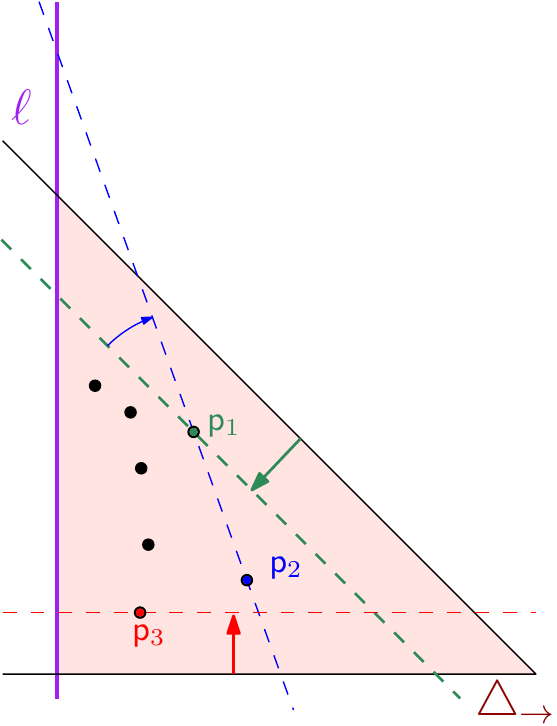}%
       \vspace{-0.3cm}%
       \caption{Handling $\TriR$.}%
       \figlab{triangle:right}%
   \end{minipage}}%

For any $\Tri$, we know that $\TriR$ will be a homothet of $\Tri$,
whose vertical edge lies on $\Line$, see \figref{triangle:cut}.  We
now transform $\TriR$ uniquely such that two points of $\PntSet$ lie
on its hypotenuse (or one point and the hypotenuse is at an angle of
$-46^\circ$) and one point of $\PntSet$ lies on its bottom edge.
Start by translating the hypotenuse towards the lower left corner of
$\TriR$ (while clipping it to $\TriR$) until it hits a point,
$\pnt_1$.  Next rotate the hypotenuse clockwise around $\pnt_1$ until
it hits a second point $\pnt_2$, or its orientation is $-46^\circ$ (as
we rotate we modify its length so that one endpoint of the hypotenuse
stays on $\Line$ and the other on the base of $\TriR$).  Next
translate the base of $\TriR$ straight upwards (while clipping it and
the hypotenuse as to maintain a right triangle) until it hits a third
point $\pnt_3$ (which may be the same as the rightmost point out of
$\pnt_1$ and $\pnt_2$).  Observe that the resulting region has the
same intersection with $\PntSet$ as $\TriR$ (except maybe for the
points on the boundary). See \figref{triangle:right}.

\parpic[r]{%
   \begin{minipage}{0.3\linewidth}%
       \includegraphics[width=0.98\linewidth]{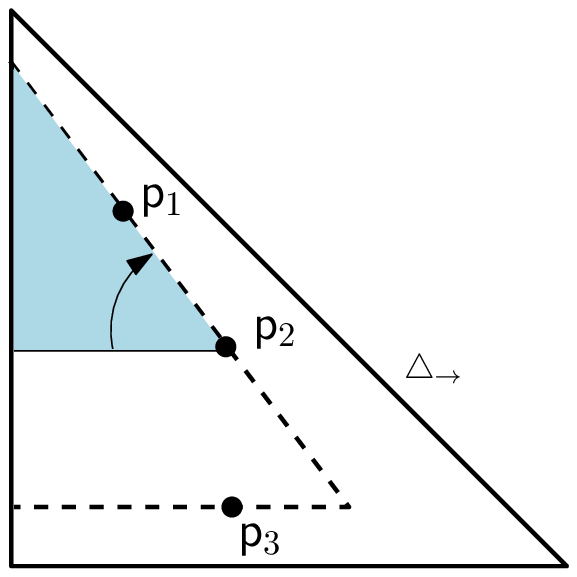}%
       \vspace{-0.3cm} \caption{}%
       \figlab{t:r:b}
   \end{minipage}}

We now bound the number of such resulting regions.  Assume that
$\pnt_2$ lies to the right of $\pnt_1$ (the other case is handled
similarly).  There are $n_v = \cardin{\PntSet_v}$ possible choices for
$\pnt_2$.  Now consider the horizontal line segment that connects
$\Line$ and $\pnt_2$.  Rotate this segment clockwise around $\pnt_2$
(while increasing its length so that the other endpoint stays on
$\Line$) until it hits $\pnt_1$, see \figref{t:r:b}.  We know that all
the points we hit in this sweeping process lie in the computed region,
and hence we can only have swept over $k$ points before reaching
$\pnt_1$ (i.e. given $\pnt_2$ there are at most $k$ choices for
$\pnt_1$.  If $\pnt_2$ does not exist we start with the triangle
formed by $\Line$ and a horizontal and $-46^\circ$ line through
$\pnt_1$).  Now imagine translating the horizontal segment connecting
$\pnt_2$ and $\Line$ straight downward till we hit $\pnt_3$ (while
increasing its length so that its right endpoint stays on the
hypotenuse defined by $\pnt_1$ and $\pnt_2$).  Again we know that all
the points we hit in this sweeping process must be in our canonical
region, and hence we can only have swept over $k$ points before
reaching $\pnt_3$ (i.e. given $\pnt_2$ and $\pnt_1$, there are at most
$k$ choices for $\pnt_3$).

Hence there are $O\pth{ n_v k^2 }$ such canonical regions for the node
$v$.  Since $\PntSet_{v_i}\cap \PntSet_{v_j} =\emptyset$ for any $v_i,
v_j$ at the same level in the top layer tree, summing across a given
level gives $O(nk^2)$ canonical regions, where $n = \cardin{\PntSet}$.
Thus, summing over all nodes in the top layer tree gives $O\pth{ n k^2
   \log n }$ such canonical regions overall.

%

\subsubsection{Handling the top left portion of the triangle
   ($\TriLT$)}

Here we must consider two cases, based on the possible locations of
$\pnt$.  If $\pnt \notin \TriLT$ (see \figref{triangle:cut}), we have
a homothet of $\Tri$ whose bottom edge lies on $\LineA$, and therefore
we can argue as in the $\TriR$ case, that this gives rise to $O\pth{
   n_{u,v} k^2 }$ different canonical regions, where $n_{u,v} =
\cardin{\PntSet_{u,v}}$.  Summing over all possible nodes $u$ and $v$
gives $O\pth{ n k^2 \log^2 n}$ such canonical regions overall.

\parpic[r]{%
   \begin{minipage}{0.25\linewidth}%
       \includegraphics[width=0.98\linewidth]{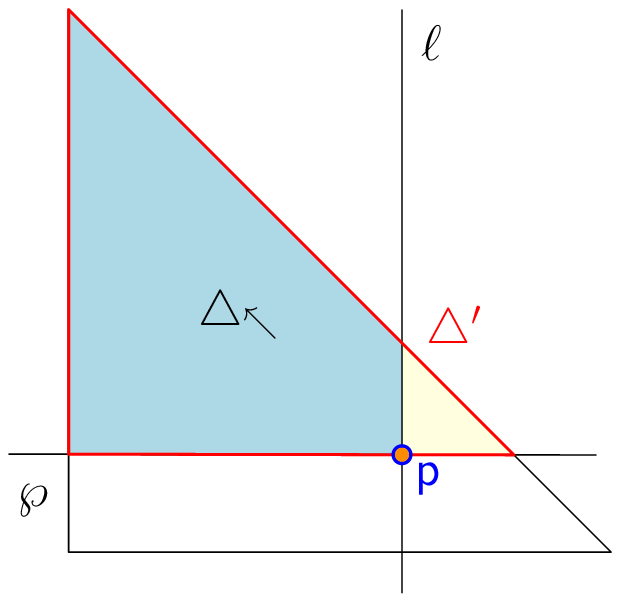}%
       \vspace{-0.3cm} \caption{}%
       \figlab{t:a:b}
   \end{minipage}}

Now suppose that $\pnt \in \TriLT$, see \figref{t:a:b}.  In this case,
we can extend $\TriLT$ to get a homothet of $\Tri$ whose right side
was cut off by $l$ in order to get $\TriLT$ ($\Tri'$ in
\figref{t:a:b}).  Clearly, we have that $\PntSet_{u,v} \cap \TriLT =
\PntSet_{u,v} \cap \Tri'$. We can now generate a canonical region for
$\Tri'$ in a similar fashion as the $\TriR$ case, since it is just a
homothet of $\Tri$ with its base lying on $\LineA$, and then we can
cut off the portion to the right of $l$.  This would imply that we can
generate $O\pth{n_{u,v} k^2}$ such canonical regions for the nodes $u$
and $v$, and so overall there are $O\pth{ n k^2 \log^2 n}$ such
canonical regions.

\subsubsection{Handling the bottom left portion of the triangle
   ($\TriLB$)}

\vspace{-1cm}
\parpic[r]{%
   \begin{minipage}{0.2\linewidth}%
       \includegraphics{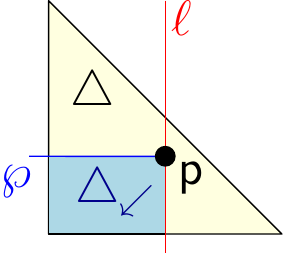}%
       \vspace{-0.3cm} \caption{}%
       \figlab{t:b:a} 
       \vspace{-0.51cm}
   \end{minipage}}

\vspace{1cm}

Again we consider two cases, based on the possible locations of
$\pnt$.  If $\pnt \in \TriLB$ (see \figref{t:b:a}), then $\TriLB$ is
an axis parallel rectangle such that one of its sides lies on $\Line$
(and another side lies on $\LineA$).  Hence by the proof of
\lemref{rect:set}, in this case $\TriLB$ gives rise to $O\pth{ k^2 n
   \log n}$ canonical regions overall.

\parpic[r]{%
   \begin{minipage}{0.25\linewidth}%
       \includegraphics{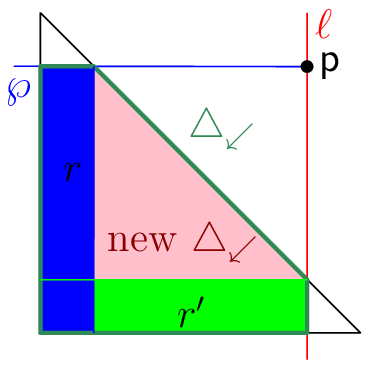}%
       \vspace{-0.6cm} \caption{}%
       \figlab{remove} 
   \end{minipage}}

Now we consider (what is by far) the hardest case, when $\pnt \notin
\TriLB$.  In order to handle this case we will need to break up
$\TriLB$ as follows.
Observe that $\TriLB$ is a rectangular region whose upper right corner
was cut off by the hypotenuse of $\Tri$.  First, we reduce $\TriLB$
into a homothet of $\Tri$, by removing rectangles $r$ and $r'$ from
the left and bottom parts of $\TriLB$, respectively (see
\figref{remove}).  This can be done since we already observed that by
the proof of \lemref{rect:set} we can construct a set of $O\pth{ nk^2
   \log^2 n}$ canonical rectangles such that any rectangle (with a
side on one of the split lines) has the same intersection with
$\PntSet$ as one of the canonical rectangles. For simplicity we
continue to refer to the remaining part of $\TriLB$ as just $\TriLB$.
\remove{
   \begin{figure}
       \TwoFigures{scale=.6}{figs/t_b_a}{(a)}
       {scale=.6}{figs/t_b_b}{(b)}
       \TwoFigures{scale=.6}{figs/remove}{(c)}
       {scale=.6}{figs/decomposition}{(d)}
       \caption{(a) $\TriLB$ contains $\pnt$.  (b) $\TriLB$ does not
          contain $\pnt$.  (c) The removed regions $r$ and $r'$.  (d)
          The decomposition of $\TriLB$ into $t_b^+$, $t_b^0$, and
          $t_b^-$.}
       
       \figlab{t_b}
   \end{figure}
}

\parpic[l]{%
   \begin{minipage}{0.24\linewidth}%
       \includegraphics[scale=0.9]{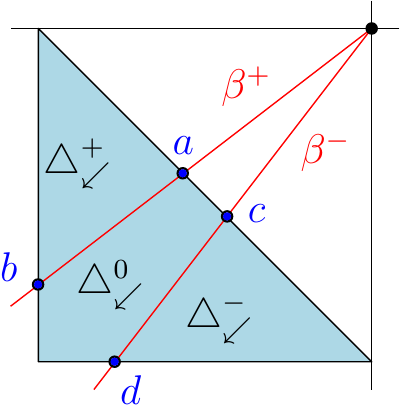}%
       \vspace{-0.4cm}%
       \caption{}%
       \figlab{decomp} 
   \end{minipage}}

We now break up $\TriLB$ into three regions.  Let $\beta^+$ and
$\beta^-$ denote the rays emanating from $\pnt$ at angles $-140^\circ$
and $-130^\circ$, respectively (again, as measured clockwise from the
positive $x$-axis).  These two lines split $\TriLB$ into three
regions, which we will denote in their counterclockwise order as
$\TriLB^+$, $\TriLB^0$ and $\TriLB^-$ (see \figref{decomp}).  Let $a$
and $b$ denote the intersection of $\beta^+$ with the hypotenuse and
left edge of $\TriLB$, respectively.  Similarly, let $c$ and $d$
denote the intersection of $\beta^-$ with the hypotenuse and bottom
edge of $\TriLB$, respectively.

\vspace{0.3cm}
\ParPicBeforeParagraph{%
   \parpic[r]{%
      \begin{minipage}{0.35\linewidth}%
          \includegraphics[width=0.95\linewidth]{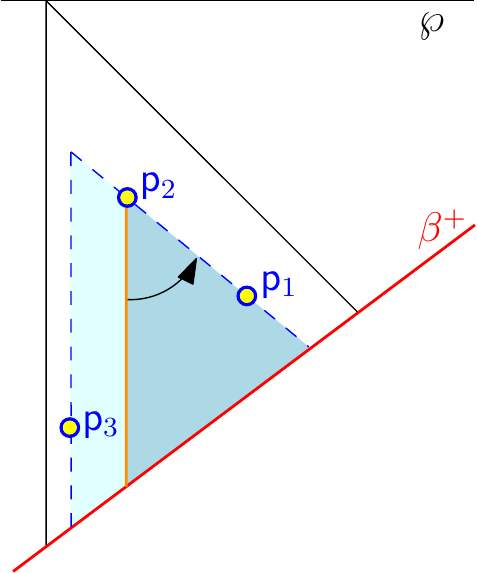}%
          \vspace{-1cm}%
          \caption{}%
          \figlab{t:bplus}%
          \vspace{0.6cm}
      \end{minipage}} }
\vspace{0.3cm}

\paragraph{Handling the top and bottom parts of $\TriLB$ (i.e.,
   $\TriLB^+$ and $\TriLB^-$).}

We now construct the canonical regions for $\TriLB^+$.  The
construction is nearly identical to that for $\TriR$ and is included
for the sake of completeness.  The construction for $\TriLB^-$ is
omitted as it is symmetric to the $\TriLB^+$ case.

Start by translating the part of the boundary that intersects the
hypotenuse of $\Tri$ towards the lower left corner of $\TriLB^+$
(while clipping it to $\TriLB^+$) until it hits a point, $\pnt_1$.
Next rotate this edge counterclockwise around $\pnt_1$ until it hits a
second point $\pnt_2$, or its orientation is $-44^\circ$ (as we rotate
we modify its length so that one endpoint stays on $\beta^+$ and the
other on the boundary $\TriLB^+$).  Next translate the vertical edge
of $\TriLB^+$ to the right (while clipping it to $\TriLB^+$) until it
hits a third point, $\pnt_3$.

As for the number of such resulting regions, assume that $\pnt_2$ lies
to the left of $\pnt_1$ (the other case is handled similarly).  There
are $n_{u,v}=\cardin{\PntSet_{u,v}}$ possible choices for $\pnt_2$.
Now consider the vertical line segment that connects $\beta^+$ and
$\pnt_2$.  Imagine rotating this segment counterclockwise around
$\pnt_2$ (while increasing its length so that the other endpoint stays
on $\beta^+$) until it hits $\pnt_1$.  We know that all the points we
hit in this sweeping process must be in our canonical region, and
hence we can only have swept over $k$ points before reaching $\pnt_1$
(if $\pnt_2$ does not exist we start with the triangle formed by
$\beta^+$ and a vertical and $-44^\circ$ line through $\pnt_1$).  Now
imagine translating the vertical segment connecting $\pnt_2$ and
$\beta^+$ to the left until we hit $\pnt_3$ (while increasing its
length so that its top endpoint stays on the line defined by $\pnt_1$
and $\pnt_2$ and its bottom endpoint on $\beta^+$).
Again we know that all the points we hit in this sweeping process must
be in our canonical region, and hence we can only have swept over $k$
points before reaching $\pnt_3$.  Hence there are $O\pth{ n_{u,v} k^2
}$ such canonical regions for a pair nodes $u$ and $v$.  Thus overall
there are $O\pth{ n k^2 \log^2 n }$ such canonical regions.

\parpic[r]{%
   \begin{minipage}{0.5\linewidth}%
       \includegraphics{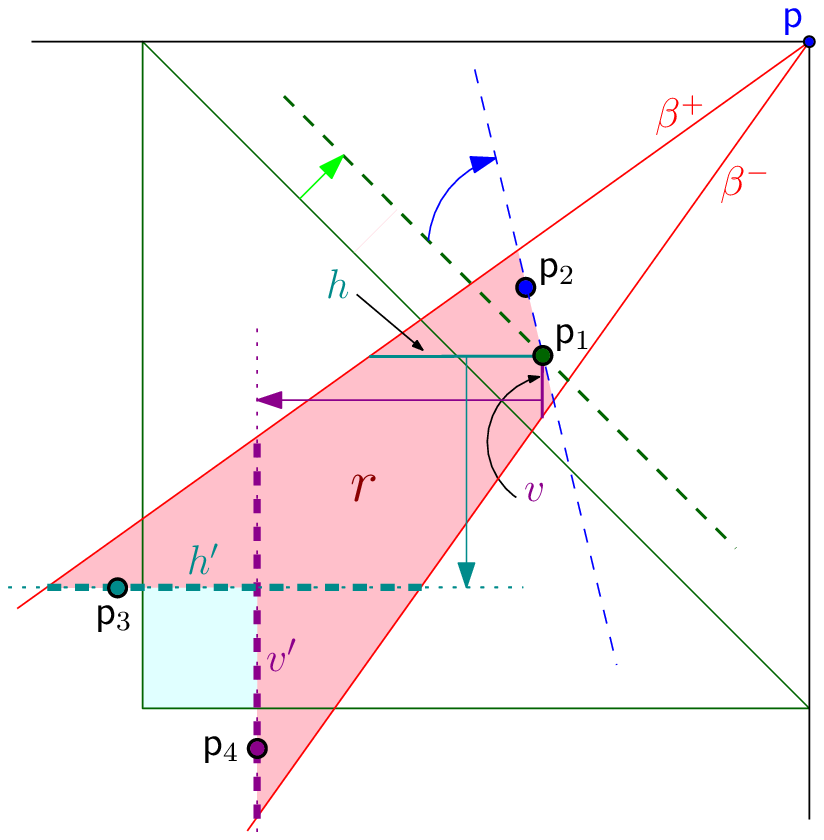}%
       \vspace{-0.8cm} %
       \caption{}%
       \figlab{triangle:middle}
   \end{minipage}}%

\paragraph{Handling the middle part of $\TriLB$ (i.e., $\TriLB^0$).}

Let $\PntSet_{u,v}^0$ denote the subset of $\PntSet_{u,v}$ that lies
in between $\beta^+$ and $\beta^-$.  Let the intersection of the
hypotenuse of $\Tri$ with $\TriLB^0$ be called the hypotenuse.

Translate the hypotenuse towards $\pnt$ (while clipping it to
$\TriLB^0$) until it hits a point $\pnt_1$.  Then rotate the
hypotenuse clockwise around $\pnt_1$ until it hits a point $\pnt_2$,
or it becomes vertical.  Without loss of generality, assume that
$\pnt_1$ lies to the right of $\pnt_2$.  Let the horizontal
(resp. vertical) line connecting $\pnt_1$ and $\beta^+$
(resp. $\beta^-$) be called $h$ (resp. $\nu$).  Translate $h$
downwards (resp. $\nu$ to the left), while enlarging it so that one
endpoint stays on $\beta^+$ (resp. $\beta^-$), until either it hits
the lowest (resp. furthest to the left) point of $\TriLB^0$ or a point
outside of $\TriLB^0$.  Let this point be denoted $\pnt_3$
(resp. $\pnt_4$), and let $h'$ (resp. $\nu '$) be the final
translation of $h$ (resp. $\nu$). See \figref{triangle:middle}.

Consider the region, $\regionA$, bounded by the portion of $h'$ to the
left of $\pnt_4$, the portion of $\nu '$ below $\pnt_3$ , $\beta^+$,
$\beta^-$, and the line going through $\pnt_1$ and $\pnt_2$ (this is
the red shaded region in \figref{triangle:middle}).  First observe
that if both $\pnt_3$ and $\pnt_4$ lie outside of $\TriLB^0$ then
$\regionA$ will not cover all the points in $\TriLB^0\cap P_{u,v}^0$.
Namely, the points lying in the rectangle defined by $h'$, $\nu '$,
and the vertical and horizontal edges of $\TriLB^0$ might not be
covered by $r$ (see \figref{triangle:middle}).  However, we already
constructed a set of $O\pth{ k^2n\log n }$ canonical rectangles, which
we know contains two canonical rectangles that cover these points, and
as such we do not have to worry about covering these points.  Clearly,
all the points of $\TriLB^0\cap P_{u,v}^0$ either lie in this
rectangle or in $\regionA$.  Next observe that there are no points of
$\PntSet_{u,v}^0$ that lie in $r$ that are not in $\TriLB\cap
P_{u,v}^0$.  This follows from the easily proven fact (i.e. tedious
but straightforward arguments) that since the hypotenuse was within
one degree of $-45^\circ$, that $h\cap \beta^+$ lies to the right of
$b$ and $\nu\cap \beta^-$ lies above $d$.

We now bound the number of canonical regions of type $r$.  There are
$\cardin{P_{u,v}^0}$ possible choices for $\pnt_1$ (which again we
assume is to the right of $\pnt_2$).  Now consider rotating $h$
clockwise around $\pnt_1$ until we hit $\pnt_2$.  We know from above
that all the points we sweep past in this process must be contained in
$\TriLB^0\cap P_{u,v}^0$ and so given $\pnt_1$ there are at most $k$
possible choices for $\pnt_2$.  Now consider translating $h$ downward
(resp. $\nu$ to the left) until we hit $\pnt_3$ (resp. $\pnt_4$).
Again, from above we know that all the points we sweep over in this
process must be contained in $\TriLB^0\cap P_{u,v}^0$ and so given
$\pnt_1$ and $\pnt_2$, there are at most $k$ possible choices for
$\pnt_3$ (resp. $\pnt_4$).  Hence there are $O\pth{ k^3
   \cardin{P_{u,v}^0} }$ such canonical regions for a given pair of
nodes $u$ and $v$, and so overall there are $O\pth{ k^3n\log^2 n }$
such canonical regions.

\subsubsection{Putting things together}
Summing the above bounds over all choices of the nodes $u$ and $v$
results overall in $O\pth{ k^3 n\log^2 n }$ canonical regions.
Furthermore, for any choice of $\Tri$, we showed above that there
exists a set of at most $\FTriConst$ of theses canonical regions whose
(union of) intersections with $\PntSet$ is the same as that of $\Tri$.
We thus get the following result.

\begin{theorem}
    Given a set $\PntSet$ of $n$ points in the plane, and parameters
    $k$ and $\alpha > 0$, one can compute a set $\TriSetA$ of $O\pth{
       k^3 n\log^2 n }$ regions, such that for any $\alpha$-fat
    triangle $\Tri$, if $\cardin{\Tri \cap \PntSet} \leq k$, then
    there exists (at most) $\FTriConst$ regions in $\TriSetA$ whose
    union has the same intersection with $\PntSet$ as $\Tri$ does.
    
    \thmlab{canonical:fat}
\end{theorem}


\section{\PTAS for Unweighted Disks and Points}
\seclab{PTAS}

In this section, we consider instances of the \PackRegions problem in
which the regions are disks with unit weights and all points have unit
capacities. We now outline a \PTAS for such instances based on the
local search technique. The algorithm and proof are an extension of
those of Chan and Har-Peled \cite{ch-aamis-11, ch-aamis-09}, and
Mustafa and Ray \cite{mr-irghs-10}.

\paragraph{The algorithm.}
Since all of the regions have unit weight, we may assume that no
region is completely contained in another. We say that a subset
$\local$ of $\ObjSet$ is \emphi{$b$-locally optimal} if $\local$ is a
pointwise independent set and one cannot obtain a larger pointwise
independent set by removing $\ell \leq b$ regions of $\local$ and
inserting $\ell + 1$ regions of $\ObjSet \setminus \local$.

Our algorithm constructs a $b$-locally optimal solution using local
search, where $b$ is some suitable constant.  We start with $\local
\leftarrow \emptyset$. We consider each subset $X \subseteq \ObjSet
\setminus \local$ of size at most $b + 1$: if $X$ is a pointwise
independent set and the set $Y\subseteq \local$ of regions pointwise
intersecting the objects of $X$ has size at most $\cardin{X}-1$, we
set $\local \leftarrow (\local \setminus Y) \cup X$. Every such swap
increases the size of $\local$ by at least one, and as such it can
happen at most $n=\cardin{\ObjSet}$ times. Therefore the running time
is bounded by $O\pth{n^{b+3} b\cardin{\PntSet}}$, since there are
$\binom{n}{b+1}$ subsets $X$ to consider and for each such subset $X$
it takes $O(n b \cardin{\PntSet})$ time to compute $Y$.

\paragraph{Analysis.}

Let $\Opt$ be the maximum pointwise independent set, and let $\local$
be the $b$-locally optimal solution returned by our algorithm.  If we
can show that the pointwise intersection graph of $\Opt\cup \local$ is
planar then the analysis in \cite{ch-aamis-11} will directly imply
that $\cardin{\local} \geq (1-O(1/\sqrt{b})) \cardin{\Opt}$.

We map the disks in $\Opt$ and $\local$ to sets of points $Q_{\Opt}$
and $Q_{\local}$ in $\Re^3$, respectively, and we map the points in
$\PntSet$ to a set of halfspaces $H_\PntSet$, by using the lifting of
disks to planes and points to rays, and then dualizing the problem
(see \secref{halfRayDisk}).  Mustafa and Ray prove that a range space
defined by a set of points and halfspaces in $\Re^3$ has the
\emphi{locality condition}, which is defined as follows.

\begin{definition}[\cite{mr-irghs-10}]
    A range space $R = (\PntSet,\EuScript{D})$ satisfies the
    \emphi{locality condition} if for any two disjoint subsets $R,B
    \subseteq \PntSet$, it is possible to construct a planar bipartite
    graph $G = (R,B,E)$ with all edges going between $R$ and $B$ such
    that for any $D\in \EuScript{D}$, if $D \cap R \neq \emptyset$ and
    $D \cap B \neq \emptyset$, then there exist two vertices $u \in D
    \cap R$ and $v \in D \cap B$ such that $(u, v) \in E$.
\end{definition}

Since $\Opt$ and $\local$ are both pointwise independent sets, we know
each point in $\PntSet$ can intersect at most one disk from $\Opt$ and
at most one disk from $\local$.  Hence each halfspace in $H_\PntSet$
can contain at most one point from $Q_\Opt$ and at most one point from
$Q_\local$.  Since points and halfspaces in $\Re^3$ have the locality
condition, setting $R=\local$ and $B=\Opt$ immediately implies that
there is a planar graph on the vertex set $\local\cup \Opt$ such that
any vertex from $\local$ and any vertex from $\Opt$ that are in the
same halfspace are adjacent.  In particular, the intersection graph is
planar.

\begin{theorem}
    Given a set of $n$ unweighted disks and a set of $m$ points in the
    plane (with unit capacities), any $b$-locally optimal pointwise
    independent set has size $\geq (1-O(1/\sqrt{b}))\Opt$, where
    $\Opt$ is the size of the maximum pointwise independent set of the
    disks.  In particular, one can compute an independent set of size
    $\geq (1-\eps)\Opt$, in time $mn^{O(1/\eps^2)}$.
    
    \thmlab{pack:uw:disks:into:u:points}
\end{theorem}

\begin{corollary}
    There is a \PTAS for instances of \PackHalfspaces in which each
    halfspace has unit weight, and each point has unit capacity.
\end{corollary}

\begin{corollary}
    There is a \PTAS for instances of \PackRegions in which each
    region is a unit-weight disk, and each point has unit capacity.
\end{corollary}

\begin{corollary}
    There is a \PTAS for instances of \PackPoints in which each region
    is a unit-capacity disk, and each point has unit weight.
\end{corollary}

\section{Hardness of approximation}
\seclab{hardness}

\subsection{Packing same size fat triangles into points}

Here we show that \PackRegions (\probref{pack:points}) does not have a
\PTAS, even if the regions have unit weight and their union complexity
is linear. We show that the problem is \APX-hard using a reduction
from the maximum bounded $3$-dimensional matching problem.  Since
maximum bounded $3$-dimensional matching is \APX-complete
\cite{k-mbtdm-91}, this will imply the claim (unless $\P=\NP$).

\begin{theorem}
    Unless $\P=\NP$ there is no \PTAS for \PackRegions
    (\probref{pack:points}) even if the regions are unweighted, in the
    plane, and have linear union complexity. In particular, this holds
    if the regions are fat triangles of similar size.  (See
    \corref{pseudo} \itemref{s:triangles} for the matching
    approximation algorithm.)
    
    \thmlab{l:b:pack:tri:into:pnts}
\end{theorem}

\begin{proof}
    Let $T\subseteq A\times B\times C$ be the input triples for an
    instance of maximum bounded $3$-dimensional matching, where $A$,
    $B$, and $C$ are disjoint subsets of some ground set $X$ (for
    simplicity we assume $X=A\cup B\cup C$).  For each element $x\in
    X$ we make a representative point $v_x$ and place it arbitrarily
    on the unit circle in the plane and give it unit capacity.  Let
    $V_A$, $V_B$, and $V_C$ be the sets of representatives for $A$,
    $B$, and $C$ (respectively).  A triple in $T$ thus corresponds to
    a triangle with one vertex in each of $V_A$, $V_B$, and $V_C$.
    Clearly, finding a maximum packing of these triangles into these
    points is an instance of \PackRegions.  Moreover, a maximum
    packing here corresponds to a maximum set of triangles (triples)
    such that each point (element of $X$) is covered by at most one
    triangle. Therefore a \PTAS for this problem translates to a \PTAS
    for the maximum bounded $3$-dimensional matching problem. (Note
    that this does not imply that there is no \PTAS for other specific
    types of regions.)
    
    \parpic[r]{\includegraphics[scale=.75]{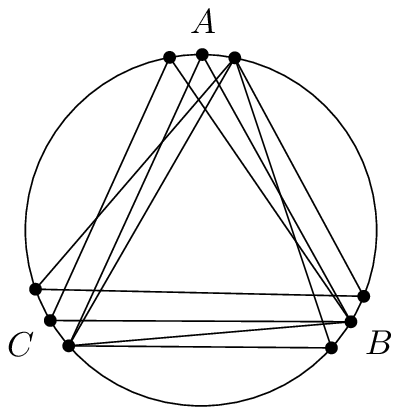}}
    Now we show that we can make the triangles fat and of similar
    size, and hence there is no \PTAS even in the case of linear union
    complexity.  Let the range of a set of representative points be
    the angle around the circle between the farthest two points of the
    set, and let the center of a set be the midpoint on the circle
    between the farthest two points of the set.  Instead of placing
    the points arbitrarily, we will place the points so that the range
    of each of $V_A$, $V_B$, or $V_C$ is less than five
    degrees. Moreover, we place the points so that the centers of
    $V_A$, $V_B$, and $V_C$ are 120 degrees apart.  In this case the
    triangles all have roughly the same size and are nearly
    equilateral.  It is known that such a set of triangles has linear
    union complexity \cite{mmpssw-ftdlm-94}. Hence, by the above
    reduction, even in this case where the regions are restricted to
    have linear union complexity (and even more specifically when they
    are restricted to be fat triangles of roughly the same size), we
    cannot get a \PTAS.
\end{proof}

\subsection{Packing points into fat triangles}

\begin{lemma}
    There is an approximation-preserving reduction from the
    \PStyle{Independent Set} problem in general graphs to the
    \PackPoints problem. In particular, for instances of the problem
    \PackPoints in which the regions are fat triangles with unit
    capacities and the points are unweighted, no approximation better
    than $\Omega( n^{1-\eps})$ is possible in polynomial time, for any
    constant $\eps > 0$, unless $\P = \NP$.
    
    \lemlab{l:b:pack:pnts:into:triangles}
\end{lemma}

\begin{proof}
    Consider an instance of the \PStyle{Independent{ }Set} problem,
    namely a graph $\graph=(V,E)$. Let $n=\cardin{V}$. Place $n$
    distinct points on the unit circle (arbitrarily) and map every
    vertex of $V$ to a unique point of the resulting set of points
    $\PntSet$.  For every edge $uv \in E$, consider the segment
    $\pnt_u \pnt_v$, where $\pnt_u$ and $\pnt_v$ are the points
    corresponding to $u$ and $v$ in $\PntSet$. We construct a fat
    triangle containing $\pnt_u\pnt_v$ by connecting $\pnt_u$,
    $\pnt_v$, and a third vertex in the interior of the unit disk;
    this can always be done so as to achieve roughly 2-fatness. We add
    this triangle to our set of regions $\ObjSet$, and assign it
    capacity one.
    
    Clearly, solving the resulting instance $(\PntSet,\ObjSet)$ of
    \PackPoints is equivalent to solving the \PStyle{Independent Set}
    problem for $G$. The claim now follows from the hardness results
    known for the \PStyle{Independent{ }Set} problem \cite{h-chaw-99}.
\end{proof}


\InFullVer{
\section{Conclusions}
\seclab{conclusions}%

In this paper, we presented a general framework for approximating
geometric packing problems with non-uniform constraints. We then
applied this framework in a systematic fashion to get improved
algorithms for specific instances of this problem, many of which
required additional non-trivial ideas. There are several special cases
of this problem for which we currently do not know any useful
approximation; for example, the special case of packing axis-parallel
boxes into points, in which the boxes are in four dimensions is still
wide open. Making some progress on these special cases is an
interesting direction for future work.

\section*{Acknowledgments}

The authors thank Timothy Chan, Chandra Chekuri, and Esther Ezra for
several useful discussions.
}


\InFullVer{
\bibliographystyle{alpha}%
\bibliography{packing}

\appendix
}
\InConfVer{%
   \section{Additional details}
   \apndlab{additional:details}%

   \subsection{Basic tools}
   \apndlab{basic:tools}
   \BasicTools%

   \ImproveRunningTime%

}

\InConfVer{%
   
  \immediate\closeout\myoutfile. 
  \section{Proofs}
  \input{fragment/myfile.tmp}
}

\end{document}